\pdfoutput=1
\documentclass[11pt]{article}

\makeatletter
\newcommand\email[2][]%
   {\newaffiltrue\let\AB@blk@and\AB@pand
      \if\relax#1\relax\def\AB@note{\AB@thenote}\else\def\AB@note{\relax}%
        \setcounter{Maxaffil}{0}\fi
      \begingroup
        \let\protect\@unexpandable@protect
        \def\thanks{\protect\thanks}\def\footnote{\protect\footnote}%
        \@temptokena=\expandafter{\AB@authors}%
        {\def\\{\protect\\\protect\Affilfont}\xdef\AB@temp{#2}}%
         \xdef\AB@authors{\the\@temptokena\AB@las\AB@au@str
         \protect\\[\affilsep]\protect\Affilfont\AB@temp}%
         \gdef\AB@las{}\gdef\AB@au@str{}%
        {\def\\{, \ignorespaces}\xdef\AB@temp{#2}}%
        \@temptokena=\expandafter{\AB@affillist}%
        \xdef\AB@affillist{\the\@temptokena \AB@affilsep
          \AB@affilnote{}\protect\Affilfont\AB@temp}%
      \endgroup
       \let\AB@affilsep\AB@affilsepx
}
\makeatother

\usepackage[hidelinks]{hyperref}

\usepackage[numbers]{natbib}

\usepackage{geometry}
\usepackage{fourier}

\usepackage{amssymb}
\usepackage{amsmath}

\usepackage{amsthm}
\usepackage{pgfplots}
\pgfplotsset{compat=1.18}
\usepackage{mathtools}
\usepackage[shortlabels]{enumitem}
\usepackage{xspace}
\newcommand\numberthis{\addtocounter{equation}{1}\tag{\theequation}}

\usepackage[noblocks]{authblk}
\renewcommand\Affilfont{\normalsize}

\newtheorem{theorem}{Theorem}[section]
\newtheorem{lemma}{Lemma}[section]

\theoremstyle{definition}
\newtheorem{definition}{Definition}[section]

\newtheorem{remark}{Remark}[section]

\usepackage[ruled,vlined,linesnumbered]{algorithm2e}

\SetCommentSty{mycommfont}

\DontPrintSemicolon

\let\oldnl\nl
\newcommand{\nonl}{\renewcommand{\nl}{\let\nl\oldnl}}
\renewcommand{\vec}[1]{\mathbf{#1}}
\DeclarePairedDelimiter{\ceil}{\lceil}{\rceil}
\DeclarePairedDelimiter{\floor}{\lfloor}{\rfloor}

\DeclareMathOperator*{\argmax}{arg\,max}
\newcommand{\set}[1]{\ensuremath{\{#1\}}}
\newcommand{\sset}[2]{\ensuremath{\{#1 \mid #2\}}}

\makeatletter
\newenvironment{mechanism}[1][htb]{%
   \begin{algorithm}[#1]%
  }{\end{algorithm}}
\makeatother

\let\OLDthebibliography\thebibliography
\renewcommand\thebibliography[1]{
  \OLDthebibliography{#1}
  \setlength{\parskip}{2pt}
  \setlength{\itemsep}{4pt plus 1pt}
}

\DeclareMathOperator{\OPTF}{OPT_F}
\newcommand{\V}{\ensuremath{v}}

\DeclareMathOperator{\OPTFk}{OPT_F^\mathit{k}}
\DeclareMathOperator{\OPTIk}{OPT_I^\mathit{k}}
\DeclareMathOperator{\OPTFn}{OPT_F^\mathit{n}}

\newcommand{\MultiMechanism}{\ensuremath{\textsc{Sort-\&-Reject$_{\alpha}^k$}}\xspace}
\newcommand{\BestInMechanism}{\ensuremath{\textsc{Greedy-Best-In$^k$}}\xspace}

\newcommand{\discretization}{\ensuremath{\textsc{Chunk-\&-Solve}}\xspace}

\newcommand{\paam}{\ensuremath{\textsc{Prune-\&-Assign}}\xspace}
\newcommand{\pruning}{\ensuremath{\textsc{Pruning}}\xspace}

\newtheorem{fact}{Fact}[section]

\begin{document}
\title{Partial Allocations in Budget-Feasible Mechanism Design: Bridging Multiple Levels of Service and Divisible Agents}

\author[1,3]{Georgios Amanatidis}
\author[2,4]{Sophie Klumper}
\author[1,3,5]{Evangelos Markakis}
\author[2,4]{\\Guido Sch\"afer}
\author[2]{Artem Tsikiridis}

\affil[1]{Athens University of Economics and Business, Greece}
\affil[2]{Centrum Wiskunde \& Informatica (CWI), The Netherlands}
\affil[3]{Archimedes, Athena Research Center, Greece}
\affil[4]{University of Amsterdam, The Netherlands}
\affil[5]{Input Output (IOG), UK}

\date{}

\maketitle              
\begin{abstract}
\noindent Budget-feasible procurement has been a major paradigm in mechanism design since its introduction by \citet{singer10}. An auctioneer (buyer) with a strict budget constraint is interested in buying goods or services from a group of strategic agents (sellers). In many scenarios it makes sense to allow the auctioneer to only partially buy what an agent offers, e.g., an agent might have multiple copies of an item to sell, they might offer multiple levels of a service, or they may be available to perform a task for any fraction of a specified time interval. Nevertheless, the focus of the related literature has been on settings where each agent's services are either fully acquired or not at all. A reason for this is that in settings with partial allocations like the ones mentioned, there are strong inapproximability results (see, e.g., \citet{chan14,anari18}). 
Under the mild assumption of being able to afford each agent entirely, we are able to circumvent such results. We design a polynomial-time, deterministic, truthful, budget-feasible, $(2+\sqrt{3})$-approximation mechanism for the setting where each agent offers multiple levels of service and the auctioneer has 
a valuation function which is separable concave, i.e., it is the sum of concave functions.
We then use this result to design a deterministic, truthful and budget-feasible $O(1)$-approximation mechanism for the setting where any fraction of a service can be acquired, again for  separable concave objectives. 
For the special case where the objective is the sum of linear valuation functions, we improve the best known approximation ratio for the problem from $(3+\sqrt{5})/2$ (by \citet{klumper22}) to $2$. This establishes a separation between this setting and its indivisible counterpart.
\end{abstract}

\section{Introduction}
Consider a procurement auction, where the agents have \emph{private} costs on the services that they can offer, and the auctioneer associates a value for each possible set of selected agents. This forms a single parameter auction environment, where the agents may strategically misreport their cost to their advantage for obtaining higher payments. Imagine now that the auctioneer additionally has a strict budget constraint that they cannot violate. Under these considerations, a natural goal for the auctioneer is to come up with a truthful mechanism for hiring a subset of the agents, that maximizes their procured value and such that the total payments to the agents respect the budget limitations. 
This is precisely the model that was originally proposed by Singer \cite{singer10} for \emph{indivisible} agents, i.e., with a binary decision to be made for each agent (hired or not). Given also that even the non-strategic version of such budget-constrained problems tend to be NP-hard, the main focus is on providing budget-feasible mechanisms that achieve approximation guarantees on the auctioneer's optimal potential value.

Since the work of \citet{singer10}, a large body of works has emerged, devoted to obtaining improved results on the original model, as well as to proposing a number of extensions. These extensions include, among others, additional feasibility constraints, richer objectives, more general valuation functions and additional assumptions, such as Bayesian modeling. Undoubtedly, all these results have significantly enhanced our understanding for the indivisible scenario. In this paper, we move away from the case of indivisible agents and concentrate on two settings that have received much less attention in the literature. In both of the models that we study, instead of hiring agents entirely or not at all, the auctioneer has more flexibility and is allowed to partially procure the services offered by each agent. We assume that the auctioneer's valuation function is the sum of individual valuation functions, each associated with a particular agent.\smallskip

 \noindent\textbf{Agents with Multiple Levels of Service:} 
In this setting, which directly generalizes Singer's setting \citep{singer10}, each agent offers a service that consists of multiple levels. We can think of the levels as corresponding to different qualities of service. Hence, the auctioneer can choose not to hire an agent, or hire the first $x$ number of levels of an agent, for some integer $x$, or hire the agent entirely, i.e., for all the levels that she is offering. Furthermore, the valuation function associated with each agent is concave, meaning that the marginal value of each level of service is non-increasing. 
This setting  was first introduced by \citet{chan14} in the context of each agent offering multiple copies of the same good and each additional copy having a smaller marginal value. In their work it was assumed that the cost of a single level is arbitrary, meaning that it is plausible that the auctioneer can only afford to hire a single level of service of a single agent. \citet{chan14} proposed randomized, truthful, and budget-feasible mechanisms for this setting, with approximation guarantees that depend on the total number of levels/copies. The crucial difference with our setting is that we assume that the auctioneer's budget is big enough to afford any single individual agent entirely, which is in line with the indivisible setting where the auctioneer can afford to hire any single agent.

 \smallskip

 \noindent\textbf{Divisible Agents:} 
A related setting is the one in which agents are offering a divisible service, e.g., offering their time. In this case, it is reasonable to assume that the auctioneer can hire each agent for any fraction of the service that they are offering. Again, the valuation function associated with each agent is assumed to be concave, meaning that the marginal gain is non-increasing in the fraction of the acquired service. 
Note that this problem is the fractional relaxation of the problem introduced by \citet{singer10}, when it is assumed that the auctioneer can afford to hire any individual divisible agent entirely. \citet{anari18} were the first to study the divisible setting. In their work they employed a \emph{large market} assumption, which, in the context of budget-feasible mechanism design, roughly means that the cost of each agent for their entire service is insignificant compared to the budget of the auctioneer. Additionally, they noticed that in the divisible setting, no truthful mechanism with a finite approximation guarantee exists without any restriction on the costs. Very recently, \citet{klumper22} revisited this problem without the large market assumption but under the much milder assumption that the auctioneer can afford to hire any individual agent entirely (which is standard in the literature for the indivisible setting, but here it does restrict the bidding space). They presented a deterministic, truthful and budget-feasible mechanism that achieves an approximation ratio of $ \frac{3+\sqrt{5}}{2}\approx 2.62$ for linear valuation functions and extended it to the setting in which all agents are associated with the same concave valuation function. 
 \smallskip

The two aforementioned settings of procurement auctions have a number of practical applications in various domains. As previously mentioned, the divisible setting would, for example, be useful to model the time availability of a worker in the context of crowdsourcing. Moreover, these types of auctions can also be applied to other industries, such as transportation and logistics, where the delivery of goods and services can be broken down into multiple levels of service. For instance, in the transportation industry, the first level of service can represent the basic delivery service, while the higher levels can represent more premium and specialized services, such as express delivery or temperature-controlled shipping. The auctioneer can then choose to hire each agent up to an available level of service, not necessarily the best offered, based on the budget constraint and the value of the services provided.

\subsection{Our Contributions}
In this work, we propose deterministic, truthful and budget-feasible mechanisms for settings with partial allocations. Specifically, 
\begin{itemize}[labelindent=\parindent,leftmargin=*,itemsep=3pt,topsep=3pt]
 \item We present a mechanism, \MultiMechanism (Mechanism \ref{mechanism:multi}), with an approximation ratio of $2+\sqrt{3} \approx 3.73$ for the indivisible agent setting with multiple levels of service and concave valuation functions (Section \ref{subsec:main_mechanism}, Theorem \ref{thm:multi-mechanism}).
 The main idea behind our novel mechanism is to apply a backwards greedy approach, in which we start from an optimal fractional solution and we discard single levels of service one by one, until a carefully chosen stopping condition is met. For this setting, no constant-factor approximation mechanism was previously known.  Note that the work of \citet{chen11} implies that no deterministic mechanism can achieve an approximation ratio better than $1 + \sqrt{2} \approx 2.41$ for this setting when $k = 1$.
 We further explore how this guarantee improves as one moves towards a \emph{large market} (Section \ref{subsec:large-market}, Theorem \ref{thm:large-market-bound}).

\item We use \MultiMechanism as a subroutine in order to design a mechanism for the setting with divisible agents, \discretization (Mechanism \ref{algo:interpolation}), that achieves an approximation ratio of $4+2\sqrt{3} \approx 7.46$ for non-decreasing concave valuation functions (Section \ref{subsec:interpolation}, Theorem \ref{thm:interpolation-thm}).
This is the first result for the problem that is independent of the number of agents $n$ and the derivative of the valuation functions close to $0$, and is a significant improvement over the conference version of this work \citep{AKMST23}.
Note that \citet{anari18} have shown that no mechanism can achieve an approximation ratio better than $e/(e-1) \approx 1.58$ for this setting.
On a technical level, we exploit the correspondence between the discrete and the continuous settings; as the number of levels of service grows arbitrarily large, the former converges to the latter. This behaviour, however, does not imply a polynomial reduction. For a modest number of levels, one could say for each discrete valuation function that it either approximates well its continuous counterpart, or that its value is concentrated in its first level.

 \item We improve on the state-of-the-art for the divisible setting with linear valuation functions \citep{klumper22} by suggesting a deterministic $2$-approxi\-mation mechanism, \paam (Mechanism \ref{mechanism:posted-price}; see Section \ref{subsec:2apx}, Theorem \ref{thm:prune-and-assign}). Note that the $e/(e-1)$ lower bound of \citet{anari18} holds even in this case, as the instance they construct has a linear valuation function. By proving this result, we establish a separation of the divisible agents model and its indivisible counterpart, for which a lower bound of $1+\sqrt{2}$ is known due to \citet{chen11}.
 Our mechanism is inspired by the \emph{randomized} $2$-approximation mechanism proposed by \citet{gravin20} for the indivisible setting. 
\end{itemize}
As we mentioned above, all our results are under the mild assumption that we can afford each agent entirely. For the setting with divisible agents this is necessary in order to achieve any non-trivial factor \citep{anari18}, and it was also assumed by \citet{klumper22}. 
Even for the discrete setting with multiple levels of service this assumption circumvents a strong lower bound of \citet{chan14} which is linear to the maximum number of levels offered by any agent
(see also Remark \ref{remark:chan_lb} and Appendix \ref{app:unrestricted_k_levels}). In both settings our assumptions are much weaker than the large market assumptions often made in the literature (see, e.g., \cite{anari18,jalaly18}). 

\subsection{Further Related Work}
The design of truthful budget-feasible mechanisms for indivisible agents was introduced by \citet{singer10}, who gave a deterministic mechanism for additive valuation functions with an approximation guarantee of $5$, along with a lower bound of $2$ for deterministic mechanisms. This guarantee was subsequently improved to $2+\sqrt{2} \approx 3.41$ by \citet{chen11}, who also provided a lower bound of $2$ for randomized mechanisms and a lower bound of $1+\sqrt{2} \approx 2.41$ for deterministic mechanisms. \citet{gravin20} gave a deterministic $3$-approximation mechanism, which is the best known guarantee for deterministic mechanisms to this day, along with a lower bound of $3$ when the guarantee is with respect to the optimal non-strategic fractional solution. Regarding randomized mechanisms, \citet{gravin20} settled the question by providing a randomized $2$-approximation mechanism, matching the lower bound of \citet{chen11}. Finally, the question has also been settled under the large market assumption by \citet{anari18}, who extended their $\frac{e}{e-1} \approx 1.58$ mechanism for the setting with divisible agents to the indivisible setting. As mentioned earlier, \citet{klumper22} studied the divisible setting without the large market assumption, but under the assumption that the private cost of each agent is bounded by the budget and gave, among other results, a deterministic $\frac{3+\sqrt{5}}{2}$-approximation mechanism for linear valuation functions. 

For indivisible agents, the problem has also been extended to richer valuation functions. This line of inquiry also started by \citet{singer10}, who gave a randomized algorithm with an approximation guarantee of $112$ for a monotone submodular objective. Once again, this result was improved by \citet{chen11} to a $7.91$ guarantee, and the same authors devised a deterministic mechanism with a $8.34$ approximation. Subsequently, the bound for randomized mechanisms was improved by Jalaly and Tardos \cite{jalaly18} to $5$. More recently, \citet{balkanski22} proposed a new method of designing mechanisms that goes beyond the sealed-bid auction paradigm. Instead, \citet{balkanski22} presented mechanisms in the form of deterministic clock auctions and, for the monotone submodular case, presented a $4.75$-approximation mechanism. Very recently, these guarantees were further improved by \citet{trieagle}, who devised a randomized and a deterministic mechanism with approximation ratios $4.3$ and $4.45$, respectively, under the same paradigm.

Beyond monotone submodular valuation functions, it becomes significantly harder to obtain truthful mechanisms with small constants as approximation guarantees. Namely, for non-monotone submodular objectives, the first randomized mechanism that runs in polynomial time was due to \citet{amanatidis19} with an approximation guarantee of $505$. This guarantee was improved to $64$ by \citet{balkanski22} who provided the first deterministic mechanism for the problem and \citet{huang23} who gave a further improvement of $(3+\sqrt{5})^2$ for randomized mechanisms. The state-of-art randomized mechanism for non-monotone submodular valuation functions is due to \citet{trieagle}, achieving an approximation ratio of $12$. 
Richer valuation functions, such as XOS valuation functions (see \citet{bei17, amanatidis17, neogi24}) and subadditive valuation functions (see \citet{dobzinski11, bei17, balkanski22, neogi24}), have been extensively studied. Until very recently, no explicit construction of a mechanism achieving a constant approximation for subadditive valuation functions was known. However, \citet{bei17} had shown that such a mechanism should exist assuming a \emph{demand query oracle}, using a non-constructive argument. Interestingly, \citet{neogi24} recently constructed such a mechanism, marking a major breakthrough.

Other settings that have been studied include environments with underlying feasibility constraints, such as downward-closed environments (\citet{amanatidis16, huang23}) and matroid constraints (\citet{leonardi17}). Other environments in which the auctioneer wants to get a set of heterogeneous tasks done and each task requires that the hired agent has a certain skill, have been studied as well, see \citet{goel14, jalaly18}. Recently, \citet{li22} studied facility location problems under the lens of budget-feasibility, in which facilities have private facility-opening costs. Finally, the problem has been studied in a beyond worst-case analysis setting by \citet{rubinstein23}.

For further details, we refer the interested reader to the very recent survey of \citet{liu24}.

\section{Model and Preliminaries}

We first define the standard budget-feasible mechanism design model below which constitutes the basis of the more general models considered in this paper. The multiple levels of service model is introduced in Section~\ref{subsec:multi-def} and the divisible agent model in Section~\ref{subsec:divisible-def}.

\subsection{Basic Model} 

We consider a procurement auction consisting of a set of agents $N = \{1, \dots, n\}$ and an auctioneer who has an available budget $B \in \mathbb{R}_{> 0}$. Each agent $i \in N$ offers a service and has a private cost parameter $c_i \in \mathbb{R}_{> 0}$, representing their true cost for providing this service in full. The auctioneer derives some value $v_i \in \mathbb{R}_{\ge 0}$ from the service of agent $i$ which is assumed to be public information.

A deterministic mechanism $\mathcal{M}$ in this setting consists of an allocation rule $\vec{x}: \mathbb{R}_{\geq 0}^n \to \mathbb{R}_{\geq 0}^n$ and a payment rule $\vec{p}: \mathbb{R}_{\geq 0}^n \to \mathbb{R}_{\geq 0}^n$. To begin with, the auctioneer collects a profile $\vec{b} = (b_i)_{i \in N} \in \mathbb{R}_{\ge 0}^n$ of declared costs from the agents. Here, $b_i$ denotes the cost declared by agent $i \in N$, which may differ from their true cost $c_i$. Given the declarations, the auctioneer determines an allocation (hiring scheme) $\vec{x}(\vec{b})=(x_1(\vec{b}),\dots, x_n(\vec{b}))$, where $x_i(\vec{b}) \in \mathbb{R}_{\geq 0}$ is the allocation decision for agent $i$, i.e., to what extent agent $i$ is hired. 
Generally, we distinguish between the \emph{divisible} and \emph{indivisible} agent setting by means of the corresponding allocation rule. In the divisible setting, each agent $i$ can be allocated fractionally, i.e., $x_i(\vec{b}) \in \mathbb{R}_{\geq 0}$. In the indivisible setting, each agent $i$ can only be allocated integrally, i.e., $x_i(\vec{b}) \in \mathbb{N}_{\ge 0}$. 
Given a (possibly fractional) allocation $\vec{x}$, we define $W(\vec{x})=\{i \in N \mid x_{i} > 0\}$ as the set of agents who are positively allocated under $\vec{x}$. 
The auctioneer also determines a vector of payments $\vec{p}(\vec{b})=(p_1(\vec{b}),\dots, p_n(\vec{b}))$, where $p_i(\vec{b})$ is the payment agent $i$ will receive for their service.

We assume that agents have quasi-linear utilities, i.e., for a deterministic mechanism $\mathcal{M}=(\vec{x}, \vec{p})$, the utility of agent $i \in N$ for a profile $\vec{b}$ is $u_i(\vec{b})= p_i(\vec{b})-c_i \cdot x_i(\vec{b})$. We are interested in mechanisms that satisfy three properties for any true profile $\vec{c}$ and any declared profile $\vec{b}$:

\begin{itemize}[labelindent=\parindent,leftmargin=*,itemsep=3pt,topsep=3pt]
\item \emph{Individual rationality:} Each agent $i \in N$ receives non-negative utility, i.e., $u_i(\vec{b}) \geq 0$.
\item \emph{Budget-feasibility:} The sum of all payments made by the auctioneer does not exceed the budget, i.e., $\sum_{i \in N}p_i(\vec{b}) \le B$.
\item \emph{Truthfulness:} Each agent $i \in N$ does not have and incentive to misreport their true cost, regardless of the declarations of the other agents, i.e., $u_i(c_i, \vec{b}_{-i}) \geq u_i(\vec{b})$ for any $b_i$ and $\vec{b}_{-i}$.
\end{itemize}

Given an allocation $\vec{x}$, the total value that the auctioneer obtains is denoted by $\V(\vec{x})$. The exact form of this function depends on the respective model we are studying and will be defined in the subsections below. 

All the models that are studied in this paper are single-parameter settings and so the characterization of \citet{myerson81} applies.\footnote{\ We refer the reader to \cite{apt22} for a rigorous treatment of the uniqueness property of Myerson's characterization result.} It is therefore sufficient to focus on the class of mechanisms with \emph{monotone non-increasing} (called \emph{monotone} for short) allocation rules. An allocation rule is monotone non-increasing if for every agent $i \in N$, every profile $\vec{b}$, and all $b_i' \leq b_i$, it holds that $x_i(\vec{b}) \leq x_i(b_i', \vec{b}_{-i})$. We will use this together with Theorem \ref{thm:archerPayments} below to design truthful mechanisms.

\begin{theorem}[\cite{archer01,myerson81}] \label{thm:archerPayments}
A monotone non-increasing allocation rule $\vec{x}(\mathbf{b})$ admits a payment rule that is truthful and individually rational if and only if for all agents $i \in N$ and all bid profiles $\mathbf{b}_{-i}$, we have $\int_{0}^{\infty} x_{i}(z,\mathbf{b}_{-i}) dz < \infty$. In this case, we can take the payment rule $p(\mathbf{b})$ to be
\begin{equation}
 \label{eq:payment-id-BF}
 p_i(\vec{b}) = b_ix_i(\vec{b}) + \int_{b_i}^{\infty}x_i(z, \vec{b}_{-i})dz \,.
\end{equation}
\end{theorem}
In this paper, we will exclusively derive monotone allocation rules that are implemented with the payment rule as defined in \eqref{eq:payment-id-BF}. Therefore, in the remainder of this paper, we adopt the convention of referring to the true cost profile $\vec{c}$ of the agents as input (rather than distinguishing it from the declared cost profile $\vec{b}$), yet we will omit the explicit reference to $\vec{c}$ if it is clear from the context. Finally, whenever tie-breaking is needed for any of our mechanisms, we assume a lexicographic precedence.

\subsection{\texorpdfstring{$k$}{k}-Level Model} 
\label{subsec:multi-def}

We consider the following multiple levels of service model as a natural extension of the standard model introduced above (see also \cite{chan14}). 
Throughout the paper, we refer to this model as the \emph{$k$-level model} for short: 
Suppose each agent $i \in N$ offers $k \ge 1$ levels of service and has an associated valuation function $v_i: \{0, \dots, k\} \to \mathbb{R}_{\ge 0}$ which is public information.\footnote{\ Our results very easily extend to the setting where there is a different (public) $k_i$ associated with each agent $i$. We use a common $k$ for the sake of presentation.} Here, $v_i(j)$ denotes the auctioneer's value for the first $j$ levels of service of agent $i$. 
Observe that in this setting each agent $i \in N$ is indivisible and the range of the allocation rule is constrained to $\set{0, \dots, k}$, i.e., $x_i: \mathbb{R}^n_{\geq 0} \to \set{0, \dots, k}$. Note also that the total cost of agent $i$ is linear (as defined above), i.e., the cost of using $x_i = j$ levels of service of agent $i$ is $j\cdot c_i$.

 \smallskip\noindent
 {\textbf{Valuation Functions:}} 

Without loss of generality, we assume that each $v_i$ is normalized such that $v_i(0) = 0$.
We study the general class of \emph{concave} valuation functions, i.e., for each agent $i$, $v_i(j)-v_i(j-1) \geq v_i(j+1)- v_{i}(j)$ for all $j=1,\dots, k-1$. We also define the \emph{$j$-th marginal valuation} of agent $i$ as $m_i(j) := v_i(j)-v_i(j-1)$, for $j \in \{ 1, \dots, k\}$, with the convention that $m_i(0) = 0$. 
Given a profile $\vec{c}$, the total value that the auctioneer derives from an allocation $\vec{x}$ is defined by the separable concave function $\V(\vec{x}(\vec{c}))=\sum_{i \in N}v_i(x_i(\vec{c}))$.

 \smallskip\noindent
 \textbf{Cost Restrictions:}
We consider different assumptions with respect to the ability of the auctioneer to hire multiple levels of service. In the \emph{all-in setting}, we assume that the auctioneer can afford to hire all levels of each single agent, i.e., given a cost profile $\vec{c}$, for every agent $i \in N$ it holds that $k \cdot c_i \le B$. Note that under this assumption we implicitly constrain the space of the (declared) cost profiles.\footnote{\ That is, we assume that any agent who violates the respective condition is discarded up front from further considerations, e.g., by running a pre-processing step that removes such agents.} In contrast, in the \emph{best-in setting}, which is equivalent to the setting of \citet{chan14}, the auctioneer is guaranteed only to be able to afford the first level of service, i.e., given a cost profile $\vec{c}$, for every agent $i \in N$ it holds that $c_i \le B$.
We focus on the all-in setting throughout this work with the only exception being Appendix~\ref{app:unrestricted_k_levels}, where we derive an almost tight result on the best-in setting.

\begin{remark}
 \label{remark:chan_lb}
 For the best-in setting \citet{chan14} show a lower bound of $k$
 for the approximation guarantee of any deterministic, truthful, budget-feasible mechanism and a lower bound of $\ln(k)$ for the approximation guarantee of any randomized, universally truthful, budget-feasible mechanism. 
 In Appendix \ref{app:unrestricted_k_levels}
 we present a $(k+2+o(1))$-approximation mechanism, named \BestInMechanism (Mechanism \ref{mechanism:bestin}), almost settling the deterministic case; although the best-in setting is not central in our work, this might be of independent interest. Note that the randomized mechanism suggested in \citet{chan14} has an approximation ratio of $4(1+\ln(nk))$.
\end{remark}

 \noindent
 \textbf{Benchmark:}
The performance of a mechanism is measured by comparing $\V(\vec{x}(\vec{c}))$ with the underlying (non-strategic) combinatorial optimization problem, which is commonly referred to as the \emph{$k$-Bounded Knapsack Problem} 
(see, e.g., \cite{martello90} for a classification of knapsack problems):
\begin{equation}
 \label{eq:BkKP}
 \OPTIk(\vec{c}) := 
 \max \sum_{i=1}^n v_i(x_i),\ 
 \text{ s.t. }
 \sum_{i=1}^n c_ix_i \leq B,\ \ 
 x_i \in \{0, \dots, k\}\,, \forall i \in N.
\end{equation}
The $k$-Bounded Knapsack Problem is NP-hard in general, since for $k=1$ it reduces to the well-known \emph{0-1 Knapsack Problem}. We say that a mechanism $\mathcal{M}=(\vec{x}, \vec{p})$ is an \emph{$\alpha$-approximation mechanism} with $\alpha \ge 1$ if 
$\alpha \cdot \V(\vec{x}(\vec{c})) \geq \OPTIk(\vec{c})$.
We also consider the relaxation of the above problem as a proxy for $\OPTIk(\vec{c})$. The definition and further details about this are deferred to Section~\ref{sec:fBkKP} below. 

An instance $I$ of the $k$-level model will be denoted by a tuple $I = (N, \vec{c}, B, k, (v_i)_{i \in N})$. Whenever part of the input is clear from the context, we omit its explicit reference for conciseness (e.g., often we refer to an instance simply by its corresponding cost vector $\vec{c}$).

\subsection{Divisible Agent Model}
\label{subsec:divisible-def}

Next, we introduce the fractional model that we study in this work. Throughout the paper, we refer to it as the \emph{divisible agent model}: 
Here the auctioneer is allowed to hire each agent for an arbitrary fraction of the full service. More precisely, each agent $i \in N$ is divisible and the range of the allocation rule is constrained to $[0,1]$, i.e., $x_i: \mathbb{R}_{\geq 0}^n \to [0,1]$. Each agent $i \in N$ has an associated valuation function $\bar{v}_i: [0,1] \to \mathbb{R}_{\geq 0}$ (which is public information), where $\bar{v}_i(x)$ represents how valuable a fraction $x \in [0,1]$ of the service of agent $i$ is to the auctioneer.

 \smallskip\noindent
 {\textbf{Valuation Functions:}} 
Also here, we assume without loss of generality that each $v_i$ is normalized such that $\bar{v}_i(0) = 0$. We focus on the general class of non-decreasing and concave valuation functions. The total value that the auctioneer derives from an allocation $\vec{x}(\vec{c})$ is defined as $\V(\vec{x}(\vec{c}))= \sum_{i \in N} \bar{v}_i (x_i(\vec{c}))$.

 \smallskip\noindent
 \textbf{Cost Restrictions:}
We assume that the auctioneer can afford each agent to the full extent. 
More formally, given a cost profile $\vec{c}$ it must hold that for each agent $i \in N$, $c_i \le B$. The observation about the cost restrictions in Section \ref{subsec:multi-def} applies to this assumption as well, i.e., we do constrain the bidding space.

 \smallskip\noindent
 \textbf{Benchmark:}
As above, the performance of a mechanism is measured by comparing $\V(\vec{x}(\vec{c}))$ with the underlying (non-strategic) combinatorial optimization problem, which we refer to as the \emph{Fractional Concave Knapsack Problem}: 
\begin{equation}
 \label{eq:fKP}
 \OPTF(\vec{c}) := 
 \max \ \sum_{i=1}^n \bar{v}_{i}(x_i) 
 \ \ \text{s.t.} \ \ 
 \sum_{i=1}^n c_i x_i \leq B, \ \ x_i \in [0,1] \; \forall i \in N. 
\end{equation}
In the divisible agent model, a mechanism $\mathcal{M}=(\vec{x}, \vec{p})$ is an $\alpha$-approximation mechanism with $\alpha \ge 1$ if $\alpha \cdot \V(\vec{x}(\vec{c})) \geq \OPTF(\vec{c})$.

An instance $I$ of the divisible agent model will be denoted by a tuple $I = (N, \vec{c}, B, (\bar{v}_i)_{i \in N})$. As mentioned before, we will omit the explicit reference of certain input parameters if they are clear from the context.

\subsection{Fractional \texorpdfstring{$k$}{k}-Bounded Knapsack Problem} 
\label{sec:fBkKP}

We also consider the \emph{Fractional $k$-Bounded Knapsack Problem} that follows from the $k$-Bounded Knapsack Problem defined in \eqref{eq:BkKP} by relaxing the integrality constraint: 
\begin{equation}
 \label{eq:fkBKP}
 \OPTFk(\vec{c}) := 
 \max \ \sum_{i=1}^n v_i(\floor{x_i}) + m_i(\ceil{x_i})(x_i-\floor{x_i}) 
 \ \ \text{s.t.} \ \
 \sum_{i=1}^nc_ix_i \leq B, \ \ x_i \in [0,k] \; \forall i \in N.
\end{equation}
Naturally, it holds that $\OPTFk(\vec{c}) \geq \OPTIk(\vec{c})$.
Note that $\mathrm{OPT_F^1}(\vec{c}) = \OPTF(\vec{c})$.

It is not hard to see that due to the fact that the objective is separable concave, $\OPTFk(\vec{c})$ inherits the well-known properties of its one-dimensional analogue. This includes the fact that an optimal solution can be computed by an adaptation of the standard greedy algorithm that sorts the elements in decreasing marginal density (marginal value per cost) and selects as many as possible \citep{hochbaum}. 
To the best of our knowledge, separable concave functions are the \textit{widest} class of objectives for which the Fractional $k$-Bounded Knapsack Problem admits optimal solutions with this structure which we crucially  utilize here.
For completeness, we state the algorithm as Algorithm \ref{algo:fk-knapsack} below, along with an easy fact that we will repeatedly use.

\begin{algorithm}[ht]
\caption{A Greedy Algorithm for Fractional $k$-Bounded Knapsack}
\nonl $\hspace{-2.3ex}\rhd$ {\bf{Input:}} An instance $I=(N, \vec{c},B, k, (v_i)_{i \in N})$ of the $k$-level model.\\
\label{algo:fk-knapsack}
Initialize an empty list $\mathcal{L}$ of $k \cdot n$ elements.

\For{$i \in N$}
{
\For{$j\in \{1,\dots, k\}$}{
Add $\frac{m_i(j)}{c_i}$ to the list $\mathcal{L}$.
}
}

Sort $\mathcal{L}$ in decreasing order. \tcp*{we break ties lexicographicaly with respect to the agents' indices}

Let $\vec{x^*}=\vec{0}$ and $j=1$.

\While{ $\sum_{i \in N}c_ix^*_i < B$ }
{
Let $\ell \in [n]$ be the agent corresponding to the $j$-th marginal value-per-cost in $\mathcal{L}$.

Set $x^*_{\ell}= x^*_{\ell}+\min\left(\frac{B-\sum_{i\in N}c_ix^*_i}{c_{\ell}}, 1 \right)$.

Set $j=j+1$.
}

Return $\vec{x^*}$.
\end{algorithm}

\begin{fact}
\label{fact:opt-knapsack-fact}
 Given an instance $I = (N, \vec{c}, B, k, (v_i)_{i \in N})$ of the Fractional $k$-Bounded Knapsack Problem, Algorithm \ref{algo:fk-knapsack} computes in time $O(kn\log (kn))$ an optimal solution $\vec{x}^*$ that has at most one coordinate with a non-integral value. 
\end{fact}

\section{A Budget-Feasible Mechanism for Multiple Levels of Service}
\label{sec:multiple}

We derive a natural truthful and budget-feasible greedy mechanism for the $k$-level model. This mechanism will also be used in our \discretization mechanism for the divisible agent model (see Section~\ref{subsec:interpolation}). Our mechanism falls into the family of mechanisms which truncate a greedy optimal solution, as is often the case in this literature (see Further Related Work).

\subsection{A Truthful Greedy Mechanism}
\label{subsec:main_mechanism}
The main idea underlying our mechanism is as follows: 
If there is an agent $i^*$ whose maximum value $v_{i^*}(k)$ is high enough (in a certain sense), then we simply pick all levels of service of this agent. Otherwise, we compute an allocation using the following greedy procedure: We first compute an optimal allocation $\vec{x}^*\!(\vec{c})$ to the corresponding Fractional $k$-Bounded Knapsack Problem (which can be done in polynomial time) and use the integral part of this solution as an initial allocation. The intuition is that this allocation is close to the optimal fractional solution because $\vec{x}^*\!(\vec{c})$ has at most one fractional component (Fact~\ref{fact:opt-knapsack-fact}). We then repeatedly discard the worst level of service (in terms of marginal value-per-cost) of an agent from this allocation, until the total value of our allocation would drop below an $\alpha$-fraction of the optimal solution. 
A similar approach was employed by \citet{klumper22}, albeit without the subtleties of our discrete setting, as they were working on the divisible agent model.

We need some more notation for the formal description of our mechanism: Given an allocation $\vec{x}$, we denote by $\ell(\vec{x})$ the agent whose $x_{\ell(\vec{x})}$-th level of service is the least valuable in $\vec{x}$, in terms of their marginal value-per-cost ratio. Notice that due to the fact that the valuation functions are concave, the worst case marginal value-per-cost ratio indeed corresponds to the $x_{\ell(\vec{x})}$-th ratio of agent $\ell(\vec{x})$. When $\vec{x}$ is clear from the context, we refer to this agent simply as $\ell$. 
A detailed description of our greedy mechanism is given in Mechanism~\ref{mech:greedy}.

\begin{mechanism}[ht]
\caption{\MultiMechanism\label{mech:greedy}}
\nonl $\hspace{-2.3ex}\rhd$ {\bf{Input:}} An instance $I=(N, \vec{c},B, k, (v_i)_{i \in N})$ of the $k$-level model.\;
\label{mechanism:multi}
Let $i^*\in \argmax_{i \in N} v_i(k)/\OPTFk(\vec{c}_{-i})$. \label{line:i*} \tcp*{recall that ties are broken lexicographically}
\If{$v_{i^*}(k) \geq \frac{\alpha}{1-\alpha} \cdot \OPTFk(\vec{c}_{-i^*})$ \label{line:ifToCheckToPicki*}}{
 Set $x_{i^*}=k$ and $x_i=0$ for all $i \in N \setminus \{i^*\}$.
 }
\Else{
Compute an optimal fractional allocation $\vec{x}^*\!(\vec{c})$ using Algorithm \ref{algo:fk-knapsack}. \tcp*{i.e., $\vec{x}^*\!(\vec{c})$ is almost integral}

Initialize $\vec{x}= (\floor{x^*_1(\vec{c})}, \dots, \floor{x^*_n(\vec{c})} )$. \tcp*{by the construction of $\vec{x}^*\!(\vec{c})$ and the all-in assumption, $\vec{x}\neq \vec{0}$}

\For(\tcp*[f]{recall that $W(\vec{x})=\{i \in N \mid x_{i} > 0\}$}){$i \in W(\vec{x})$ }{
\For{$j \in \{1, \dots, x_i\}$}{
Add the marginal value-per-cost ratio $m_i(j)/c_i$ to a list $\mathcal{L}$.
}
}

Sort $\mathcal{L}$ in non-increasing order and let $\ell$ be the agent corresponding to the last marginal value-per-cost in $\mathcal{L}$.

\While{$\V(\vec{x}) -m_{\ell}(x_{\ell}) \geq \alpha \OPTFk(\vec{c})$ \label{line:whileSort}} 
{

Set $x_{\ell}= x_{\ell}-1$. \;

Remove the last element from $\mathcal{L}$ and update $\ell$. \;
}}
Allocate $\vec{x}$ and determine the payments $\vec{p}$ according to \eqref{eq:payment-id-BF}.
\end{mechanism}

The main result of this section is the following theorem:
\begin{theorem}
\label{thm:multi-mechanism}
 For $\alpha=\frac{1}{2+\sqrt{3}}$, \MultiMechanism is a truthful, individually rational, budget-feasible $(2+\sqrt{3})$-approximation mechanism for instances of the $k$-level model, and runs in time polynomial in $n$ and $k$.
\end{theorem}

The polynomial running time for computing the final allocation $\vec{x}$ is straightforward given Fact \ref{fact:opt-knapsack-fact}. After determining $\vec{x}$, computing the payment $p_i$ of each agent $i \in N$ can be done efficiently, as we describe in Appendix \ref{appendix:new-section}. In the remainder of this section, we prove several lemmata to establish the properties stated in Theorem~\ref{thm:multi-mechanism}. Technically, the most challenging part is to prove that the mechanism is budget-feasible (see Section~\ref{sec:BF-Greedy}).

The following property directly follows by construction of the mechanism. 
\begin{fact}
The allocation $\vec{x}$ returned by \MultiMechanism satisfies $x_i\leq x_i^*$ for every $i\in N$.
\end{fact}

We now show that the allocation rule of \MultiMechanism is monotone.

\begin{lemma}
\label{lemma:monotone}
For any $\alpha\in(0,1)$, the allocation rule of\ \ \MultiMechanism is monotone. 
\end{lemma}

\begin{proof}
Let $\vec{c}$ be a cost profile. We distinguish the following two cases:
\smallskip

\noindent\underline{Case 1:} $v_{i^*}(k) \geq \frac{\alpha}{1-\alpha} \cdot \OPTFk(\vec{c}_{-i^*})$. In this case, $i^*$ is hired for $k$ levels of service. Suppose that $i^*$ unilaterally deviates and decreases their cost to $c_{i^*}' < c_{i^*}$. Such a deviation has no influence on which agent is chosen in line \ref{line:i*} and does not alter the condition of this case. To see that, notice that the quantity $\OPTFk(\vec{c}_{-i^*})$ does not depend on the bid of $i^*$, whereas $\OPTFk\big(\big(c_{i^*}', \vec{c}_{-i^*}\big)_{-i}\big) \ge \OPTFk(\vec{c}_{-i})$, for any $i\neq i^*$. Therefore, for any such deviation, $i^*$ will remain the sole winner and will be hired for $k$ levels of service, i.e., $x_{i^*}(c_{i^*}',\vec{c}_{-i^*})=x_{i^*}(\vec{c})=k$ (where in a slight abuse of notation we introduce arguments to the components of $\vec{x}$ to distinguish between the two runs). No other agent was winning in this case, hence there is no need to examine deviations by other agents. 
\smallskip

\noindent\underline{Case 2:} $v_{i^*}(k)< \frac{\alpha}{1-\alpha} \cdot \OPTFk(\vec{c}_{-i^*})$. Here, the allocation rule of \MultiMechanism is allocating to a set $W(\vec{x}(\vec{c}))$. Note that because we are in the all-in setting, $W(\vec{x}(\vec{c}))$ starts as a non-empty set, and our \textbf{while} condition guarantees that it remains so. Fix an agent $i$ in $W(\vec{x}(\vec{c}))$ and suppose $i$ unilaterally deviates and declares $c_i' < c_i$. First of all, note that for every $j \in N\setminus\{i\}$ it holds that 
\begin{equation*}
 \frac{v_j(k)}{\OPTFk(c_i', \vec{c}_{-\{j,i\}})} \leq \frac{v_j(k)}{\OPTFk(\vec{c}_{-j})} \leq \frac{v_{i^*}(k)}{\OPTFk(\vec{c}_{-i^*})}\,,
 \end{equation*}
where the first inequality follows from the monotonicity of $\OPTFk(\cdot)$ with respect to each cost, whereas for $i$ herself we have 
\begin{equation*}
 \frac{v_i(k)}{\OPTFk(\vec{c}_{-i})} \leq \frac{v_{i^*}(k)}{\OPTFk(\vec{c}_{-i^*})}\,,
 \end{equation*}
from before she deviated.

Therefore, even if some agent $\hat{\imath} \neq i^*$ is chosen in line \ref{line:i*} under the profile $(c_i', \vec{c}_{-i})$, the mechanism will always execute the \textbf{else} case. Moreover, by the way $\vec{x}^*\!(c_i', \vec{c}_{-i})$ is constructed, the corresponding initial integral solution $\vec{x}$ in line \ref{line:init_rounding} will have an $i$-th coordinate at least as large as before.
Now consider what happens to the \textbf{while} condition of the mechanism. On the one hand, it can only be that $\OPTFk(c_i', \vec{c}_{-i}) \geq \OPTFk(\vec{c})$. On the other hand, the marginal value-per-cost ratios of agent $i$ under the profile $(c_i',\vec{c}_{-i})$ have a better position in the ordering constructed by \MultiMechanism. 
Therefore, agent $i$ will be hired to at least the same extent or more, i.e., for the final allocations we have $x_i(c_i', \vec{c}_{i}) \geq x_i(\vec{c})$, proving monotonicity. 
\end{proof}

Since the payments are computed according to \eqref{eq:payment-id-BF}, we conclude that the mechanism is truthful and individually rational. We continue by showing that \MultiMechanism achieves the claimed approximation guarantee.

\begin{lemma}
 \label{lemma:guarantee}
 Fix any $\alpha \in \big(0, \frac{3-\sqrt{5}}{2}\big)$. For the allocation $\vec{x}(\vec{c})$ computed by \MultiMechanism with input the cost profile $\vec{c}$, it holds that $\V(\vec{x}(\vec{c})) \geq \alpha \OPTIk(\vec{c})$.
\end{lemma}

\begin{proof}
For a cost profile $\vec{c}$ we will prove the claimed guarantee against the optimal value of the fractional relaxation of the bounded knapsack instance, i.e., we will show that $\V(\vec{x}(\vec{c})) \geq \alpha \OPTFk(\vec{c})$. As mentioned in Section \ref{subsec:multi-def}, this establishes our guarantee since $\OPTFk(\vec{c}) \geq \OPTIk(\vec{c})$. Again we distinguish the following two cases:
\smallskip 

\noindent\underline{Case 1:} $v_{i^*}(k) \geq \frac{\alpha}{1-\alpha} \cdot \OPTFk(\vec{c}_{-i^*})$. We directly have
\begin{equation*}
 v_{i^*}(k) \geq \frac{\alpha}{1-\alpha} \cdot \OPTFk(\vec{c}_{-i^*}) \geq \frac{\alpha}{1-\alpha} \big(\OPTFk(\vec{c})-v_{i^*}(k) \big)\,,
 \end{equation*}
 where the second inequality follows from the fact that $\OPTFk(\vec{c}_{-i}) + v_i(k) \geq \OPTFk(\vec{c})$ for all $i \in N$. Rearranging terms leads to $v_{i^*}(k) \geq \alpha \OPTFk(\vec{c})$ which concludes the case since here $v_{i^*}(k) = \V(\vec{x}(\vec{c}))$.
\smallskip

\noindent\underline{Case 2:} $v_{i^*}(k) < \frac{\alpha}{1-\alpha} \cdot \OPTFk(\vec{c}_{-i^*})$. In this case, whenever the \textbf{while} loop of \MultiMechanism runs at least once, we immediately obtain that $\V(\vec{x}(\vec{c})) \geq \alpha \OPTFk(\vec{c})$, for any $\alpha < 1$. 
Alternatively, consider an instance for which the \textbf{while} condition evaluates to \emph{False} the first time it is checked. In this case, the initial version of $\vec{x}$ cannot be equal to $\vec{x}^*\!(\vec{c})$. As a consequence, under $\vec{x}^*\!(\vec{c})$, the last level of service of some agent, say agent $f$, is fractionally included; in particular, this means that $B-\sum_{i \in N}\floor{x^*_i}c_i < c_f$. We argue that, for $\alpha \leq \frac{3-\sqrt{5}}{2}$, we still obtain $\V(\vec{x}(\vec{c})) \geq \alpha \OPTFk(\vec{c})$ and the lemma follows. 
Toward a contradiction, suppose that this is not the case. This implies that,
\begin{equation*}
 \alpha \OPTFk(\vec{c}) > \V(\vec{x}(\vec{c})) = \OPTFk(\vec{c}) - m_f(\ceil{x^*_f})\,\frac{B-\sum_{i \in N}\floor{x^*_i}c_i}{c_f} > \OPTFk(\vec{c}) - m_f(\ceil{x^*_f})\,,
\end{equation*}
which, by rearranging terms, yields 
\begin{equation}\label{ineq:m_f_lower}
m_f(\ceil{x^*_f})> (1-\alpha)\OPTFk(\vec{c})\,.
\end{equation}
At the same time, we have
\begin{equation}\label{ineq:m_f_upper}
\frac{m_f(\ceil{x^*_f})}{\OPTFk(\vec{c})} \leq \frac{v_f(k)}{\OPTFk(\vec{c}_{-f})} \leq \frac{v_{i^*}(k)}{\OPTFk(\vec{c}_{-i^*})} < \frac{\alpha}{1-\alpha}\,,
\end{equation}
where the first inequality follows from straightforward bounds on both the numerator and the denominator, the second inequality holds by the choice of $i^*$, and the third one is the very definition of Case 2.

Now, by combining inequalities \eqref{ineq:m_f_lower} and \eqref{ineq:m_f_upper} we obtain that $\alpha$ must be such that $\frac{\alpha}{1-\alpha} > 1-\alpha$, or equivalently, that $\alpha > (3-\sqrt{5})/{2}$, a contradiction. 
\end{proof}

\subsection{Ensuring the Budget-Feasibility of \texorpdfstring{\MultiMechanism}{\textsc{Sort-\&-Reject}}}
\label{sec:BF-Greedy}

It remains to prove that \MultiMechanism is budget-feasible. We introduce some auxiliary notation: 
Consider a cost profile $\vec{c}$ and an agent $i \in W(\vec{x}(\vec{c}))$. 
Let $j \in \set{1, \dots, x_i(\vec{c})}$ be any level of service among the ones for which agent $i$ is hired. Intuitively, we refer to the \emph{critical payment} $p_{ij}(\vec{c}_{-i})$ for level of service $j$ of $i$ as the largest cost $p$ that $i$ can declare and still have their level of service $j$ picked by the mechanism (see Figure~\ref{fig:payments} for an illustration). More formally, we define $Q_{ij}(\vec{c}_{-i})$ as the set of all points $q$ satisfying 
$\lim_{z \rightarrow q^-} x_i(z, \vec{c}_{-i}) \ge j$ and $\lim_{z \rightarrow q^+} x_i(z, \vec{c}_{-i}) \le j$ 
and let $p_{ij}(\vec{c}_{-i}) = \sup(Q_{ij}(\vec{c}_{-i}))$. Such a number $q$ must always exist {(e.g., $c_i \in Q_{ij}(\vec{c}_{-i}$)} and $c_i \le q \le \frac{B}{k}$). To see the latter, note that $x_i(c_i, \vec{c}_{-i}) \ge j$, which implies that $c_i \le q$, and that $x_i(z, \vec{c}_{-i}) = 0 < j$ for all $z > \frac{B}{k}$ (by our all-in assumption that would prune any agent declaring cost more than $\frac{B}{k}$), which implies that $q \le \frac{B}{k}$. Given that $Q_{ij}(\vec{c}_{-i})$ is nonempty and bounded from above, $p_{ij}(\vec{c}_{-i})$ always exists.
It is easy to see that the payment of an agent $i$ can be written as the sum over these critical payments for the levels of service $i$ was hired for.

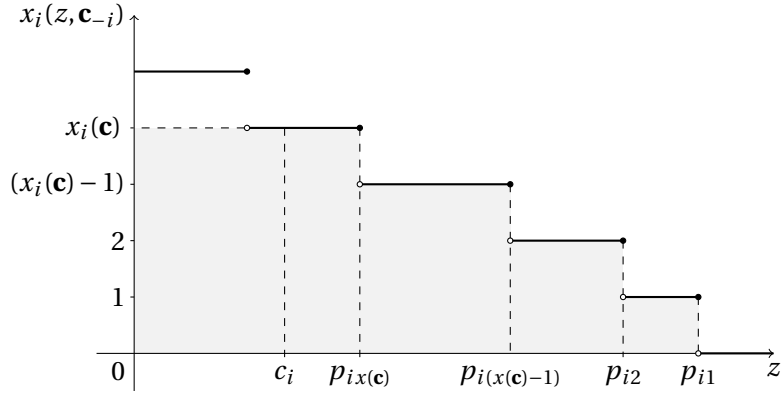
\begin{figure}[!ht]
\centering
\begin{tikzpicture}[smooth,scale=1]
\filldraw[draw=none,dashed,fill=gray!10] (0,0) rectangle (2,3); \filldraw[draw=none,dashed,fill=gray!10] (2,0) rectangle (3,3); \filldraw[draw=none,dashed, fill=gray!10] (3,0) rectangle (5, 2.25); \filldraw[draw=none,dashed, fill=gray!10] (5,0) rectangle (6.5, 1.5); \filldraw[draw=none,dashed, fill=gray!10] (6.5,0) rectangle (7.5, 0.75);
\draw[->] (-0.5,0) -- (8.5,0) node[below] {$z$};
\draw[->] (0,-0.5) -- (0,4.5) node[left] {$x_i(z,\textbf{c}_{-i})$};
\path (0,0) node[below left] {$0$}; \draw[dashed] (1.47,3) -- (-0.05,3) node[left] {$x_i(\mathbf{c})$}; \draw (-0.05,2.25) -- (0,2.25) node[left] {$(x_i(\mathbf{c})-1)$}; \draw (-0.05,1.5) -- (0,1.5) node[left] {$2$}; \draw (-0.05,0.75) -- (0,0.75) node[left] {$1$}; 
\draw[dashed] (2,3) -- (2,-0.05) node[below] {$c_{i}$}; \draw[dashed] (3,3) -- (3,-0.05) node[below] {$p_{ix(\vec{c})}$}; \draw[dashed] (5,2.25) -- (5,-0.05) node[below] {$p_{i(x(\vec{c})-1)}$} ; \draw[dashed] (6.5,1.5) -- (6.5,-0.05) node[below] {$p_{i2}$}; \draw[dashed] (7.5,0.75) -- (7.5,-0.05) node[below] {$p_{i1}$};
\draw[-,thick] (0,3.75) -- (1.5,3.75); \filldraw [] (1.5,3.75) circle (1pt); 
\filldraw [fill=white] (1.5,3) circle (1pt); \draw[-,thick] (1.53,3) -- (3,3); \filldraw [] (3,3) circle (1pt); 
\filldraw [fill=white] (3,2.25) circle (1pt); \draw[-,thick] (3.03,2.25) -- (5,2.25); \filldraw [] (5,2.25) circle (1pt); 
\filldraw [fill=white] (5,1.5) circle (1pt); \draw[-,thick] (5.03,1.5) -- (6.5,1.5); \filldraw [] (6.5,1.5) circle (1pt); 
\filldraw [fill=white] (6.5,0.75) circle (1pt); \draw[-,thick] (6.53,0.75) -- (7.5,0.75); \filldraw [] (7.5,0.75) circle (1pt); 
\filldraw [fill=white] (7.5,0) circle (1pt); \draw[-,thick] (7.53,0) -- (8.5,0); 
\end{tikzpicture}
\caption{Illustration of the critical payments of agent $i$.\label{fig:payments}}
\end{figure}

\begin{lemma}\label{lemma:payment-breakdown}
Let $\vec{c}$ be a cost profile and assume $i \in W(\vec{x}(\vec{c}))$. It holds that $\displaystyle p_i(\vec{c})= \sum_{j=1}^{x_i(\vec{c})}p_{ij}(\vec{c}_{-i}).$
\end{lemma}

\begin{proof}
Starting with \eqref{eq:payment-id-BF} from Theorem \ref{thm:archerPayments}, we obtain that
\[
 p_i(\vec{c}) = c_ix_i(\vec{c}) + \int_{c_i}^\infty \!\!\!x_i(z, \vec{c}_{-i})dz = c_ix_i(\vec{c}) + \!\!\sum_{j=1}^{x_i(\vec{c})-1}\!\!\int_{p_{i(j+1)}(\vec{c}_{-i})}^{p_{ij}(\vec{c}_{-i})} \!\!\!jdz +\int_{c_i}^{p_{ix_i(\vec{c})}(\vec{c}_{-i})} \!\!\!\!\!x_i( \vec{c})dz =\sum_{j=1}^{x_i(\vec{c})}p_{ij}(\vec{c}_{-i})\,.\qedhere
\]
\end{proof}

Lemma \ref{lemma:key} is the main technical tool needed to establish budget-feasibility for the \textbf{else} part of the \MultiMechanism mechanism. 
It is also used in the proofs of Theorems \ref{thm:large-market-bound} and \ref{thm:interpolation-thm} (implicitly in the latter) for the large market and  the divisible agent setting, respectively.

\begin{lemma}
 \label{lemma:key}
 Fix any $\alpha \in (0,1)$. Let $\vec{c}$ be a cost profile such that $v_{i^*}(k) < \frac{\alpha}{1-\alpha}\OPTFk(\vec{c}_{-i^*})$. Then,
\begin{equation*}
 \sum_{i=1}^np_i(\vec{c}) < \frac{B}{1-\alpha}\left(\frac{m_{\ell}(x_{\ell}(\vec{c}))}{\OPTFk(\vec{c}_{-\ell})} + \frac{\alpha}{1-\alpha} \right)\,.
 \end{equation*}
\end{lemma}

Observe that using Lemma \ref{lemma:key}, we can determine a range for values of $\alpha$, for which \MultiMechanism is budget-feasible:

\begin{lemma}
 \label{lemma:multi-bf}
 For $\alpha \in \Big(0, \frac{1}{2+\sqrt{3}}\Big]$, \MultiMechanism is budget-feasible.
\end{lemma}

\begin{proof}
Let $\vec{c}$ be a cost profile.
If $v_{i^*}(k) \geq \frac{\alpha}{1-\alpha} \OPTFk(\vec{c}_{-i^*})$, agent $i^*$ is the only agent hired and for $k$ levels of services. By Lemma \ref{lemma:payment-breakdown}, we obtain
\begin{equation*}
 p_{i^*}(\vec{c}) = \sum_{j=1}^k p_{i^*\!j}(\vec{c}_{-i^*})= k \cdot p_{i^*\!k}(\vec{c}_{-i^*}) \le k\,\frac{B}{k}=B\,.
\end{equation*}
The second equality follows from the definition of the critical payments and from the fact that the allocation of $i^*$ will only be smaller than $k$ when another agent becomes $i^*$ or when $c'_i > \frac{B}{k}$. The inequality follows from the fact that 
$p_{i^*k}(\vec{c}_{-i^*}) \le \frac{B}{k}$ as argued above.

Consider now the other case, i.e., a profile $\vec{c}$ such that $v_{i^*}(k) < \frac{\alpha}{1-\alpha} \OPTFk(\vec{c}_{-i^*})$. In this case, we may invoke Lemma \ref{lemma:key} and, thus, we get
\begin{align*}
 \sum_{i=1}^np_i(\vec{c}) &< \frac{B}{1-\alpha}\left(\frac{m_{\ell}(x_{\ell}(\vec{c}))}{\OPTFk(\vec{c}_{-\ell})} + \frac{\alpha}{1-\alpha} \right)\leq \frac{B}{1-\alpha}\left(\frac{v_{\ell}(k)}{\OPTFk(\vec{c}_{-\ell})} + \frac{\alpha}{1-\alpha} \right)\\
 &\leq \frac{B}{1-\alpha}\left(\frac{v_{i^*}(k)}{\OPTFk(\vec{c}_{-i^*})} + \frac{\alpha}{1-\alpha} \right) < \frac{B}{1-\alpha}\left(\frac{\alpha}{1-\alpha} + \frac{\alpha}{1-\alpha} \right) \leq B\frac{2\alpha}{(1-\alpha)^2}\,.
\end{align*}
The second inequality follows by observing that $m_{\ell}(x_{\ell}(\vec{c})) \leq v_{\ell}(k)$ since $v_{\ell}(\cdot)$ is non-decreasing. The next two inequalities are due to the definition of $i^*$ and the fact that, in this case, $v_{i^*}(k) < \frac{\alpha}{1-\alpha} \OPTFk(\vec{c}_{-i^*})$.

To obtain the budget-feasibility \MultiMechanism, we must ensure that $\frac{2\alpha}{(1-\alpha)^2} \leq 1$ or, equivalently, that $\alpha \leq \frac{1}{2+\sqrt{3}}$. The proof follows.
\end{proof}

It is clear that Lemma \ref{lemma:multi-bf} along with Lemmata \ref{lemma:monotone} and \ref{lemma:guarantee} conclude the proof of Theorem \ref{thm:multi-mechanism}.

The rest of this section is devoted to proving Lemma \ref{lemma:key}. We start by presenting a series of auxiliary statements, which will prove to be useful in our analysis. The purpose of these statements is to characterize and give upper bounds on the individual payments of winning agents, whenever the \textbf{else} part of the mechanism is executed. We begin with Lemma \ref{lemma:magic-lemma}, in which we derive an upper bound on the costs of winning agents. Observe that the final value of $\ell$ is the index of the agent with the smallest value-per-cost ratio in the allocation computed by the mechanism.

\begin{lemma}
 \label{lemma:magic-lemma}
 Fix any $\alpha \in (0, 1)$ and let $\vec{c}$ be a cost profile for an instance with $v_{i^*}(k) < \frac{\alpha}{1-\alpha} \OPTFk(\vec{c}_{-i^*})$. At the end of a run of \MultiMechanism it holds that
\begin{equation*}
 c_{\ell} < \frac{B}{1-\alpha} \cdot \frac{m_{\ell}(x_{\ell}(\vec{c}))}{\OPTFk(\vec{c})}\,.
\end{equation*}
\end{lemma}

\begin{proof}
Let $\vec{x}$ be the allocation vector at the end of a run of \MultiMechanism with input $\vec{c}$, and use $\vec{x}^*\!:=\vec{x}^*\!(\vec{c})$ for brevity. Observe that, since $v_{i^*}(k) < \frac{\alpha}{1-\alpha} \OPTFk(\vec{c}_{-i^*})$, it must be that $\V(\vec{x})-m_{\ell}(x_{\ell}) < \alpha \OPTFk(\vec{c})$, which implies that
\begin{equation}
\label{eq:a-flipped}
 \OPTFk(\vec{c})-\V(\vec{x}) + m_{\ell}(x_{\ell}) > (1-\alpha) \OPTFk(\vec{c})\,.
\end{equation}
By Fact \ref{fact:opt-knapsack-fact}, there exists at most one agent in $W(\vec{x}^*\!)$ with a non-integer allocation. We denote that agent by $f$ (if no such agent exists, let $f$ be an arbitrary agent in $W(\vec{x}^*\!)$). Note that it is possible that $f=\ell$. We have
 \begin{align*}
 B &\geq \sum_{i \in W(\vec{x}^*\!)}c_ix^*_i 
 \geq c_{\ell}(x^*_{\ell}-x_{\ell}+1) + \sum_{i \in W(\vec{x}^*\!) \setminus 
 \{\ell\}}c_i(x^*_i-x_i) \\
 &=\sum_{j=x_{\ell}}^{\lfloor x_{\ell}^* \rfloor}\frac{c_{\ell}}{m_{\ell}(j)}m_{\ell}(j) + \sum_{i \in W(\vec{x}^*\!)\setminus 
 \{\ell\}}\sum_{j=x_i+1}^{\lfloor x^*_i \rfloor}\frac{c_i}{m_i(j)}m_i(j) \ \ + \ \ \frac{c_f}{m_f(\lceil x_f^*\rceil)}(x_f^* - \lfloor x^*_f \rfloor )m_f(\lceil x_f^*\rceil) \\
 &\geq \frac{c_{\ell}}{m_{\ell}(x_{\ell})} \left( \sum_{j=x_{\ell}}^{\floor{x_{\ell}^*}}m_{\ell}(j) +\sum_{i \in W(\vec{x}^*\!)\setminus\{\ell\}}\sum_{j=x_i+1}^{\floor{x^*_i}}m_i(j) \ \ + \ \ (x_f^* - \lfloor x^*_f \rfloor )m_f(\lceil x_f^*\rceil)\right)\\
 &= \frac{c_{\ell}}{m_{\ell}(x_{\ell})} \left( m_{\ell}(x_{\ell}) +\sum_{i \in W(\vec{x}^*\!)}\sum_{j=x_i+1}^{\floor{x^*_i}}m_i(j) \ \ + \ \ (x_f^* - \lfloor x^*_f \rfloor )m_f(\lceil x_f^*\rceil)\right)\\
 &= \frac{c_{\ell}}{m_{\ell}(x_{\ell})} \left( \OPTFk(\vec{c})-\V(\vec{x}) + m_{\ell}(x_{\ell}) \right) 
 > \frac{c_{\ell}}{m_{\ell}(x_{\ell})} (1-\alpha) \OPTFk(\vec{c})\,.
 \end{align*}
The first inequality follows by the feasibility of $\vec{x}^*$, whereas the second inequality follows by the fact that $x_{\ell}\geq 1$ and, thus, every cost now has a smaller or equal coefficient than before.
The next inequality is due to the marginal value-per-cost ordering that Algorithm \ref{algo:fk-knapsack} uses to built $\vec{x}^*\!$ and \MultiMechanism uses (in the reverse order) to obtain $\vec{x}$. Finally, the last inequality is due to \eqref{eq:a-flipped}. The lemma follows by rearranging terms.
\end{proof}

We now proceed to obtaining an upper bound on the payments each agent receives for each level of service.

\begin{lemma}\label{lem:critical-upper-opt}
 Fix any $\alpha \in (0, 1)$ and let $\vec{c}$ be a cost profile such that $v_{i^*}(k) < \frac{\alpha}{1-\alpha} \OPTFk(\vec{c}_{-i^*})$. Moreover, let $i \in W(\vec{x}(\vec{c}))$. For $1\le j \le x_i(\vec{c})$, it holds that
\begin{equation*}
 p_{ij}(\vec{c}_{-i}) \leq \frac{B}{1-\alpha} \cdot \frac{m_i(j)}{\OPTFk(\vec{c}_{-i})}\,.
 \end{equation*}
\end{lemma}
\begin{proof}
Let $p\in Q_{ij}(\vec{c}_{-i})$, i.e., $p$ is a cost declaration that guarantees agent $i$ at least $j$ levels of service, and $\lambda$ be the index of the agent with the smallest value-per-cost ratio in the allocation computed by \MultiMechanism with input $(p,\vec{c}_{-i})$ (as opposed to $\ell$ when the input is $\vec{c}$).
By the definition of $p$ and $\lambda$, we have $m_i(j)/p \ge m_{\lambda}(x_{\lambda}(p, \vec{c}_{-i})) / c_{\lambda}$. By rearranging and applying Lemma \ref{lemma:magic-lemma} for the profile $(p, \vec{c}_{-i})$ and $\lambda$, we get
\begin{equation*}
p \leq \frac{c_{\lambda}m_i(j)}{m_{\lambda}(x_{\lambda}(p, \vec{c}_{-i}))} < \frac{B}{1-\alpha} \cdot \frac{m_{i}(j)}{\OPTFk(p, \vec{c}_{-i})} \leq \frac{B}{1-\alpha} \cdot \frac{m_{i}(j)}{\OPTFk(\vec{c}_{-i})}\,,
\end{equation*}
where the last inequality follows by the monotonicity of $\OPTFk(\cdot)$ with respect to subinstances.
This implies that
\[p_{ij}(\vec{c}_{-i}) = \sup(Q_{ij}(\vec{c}_{-i})) \leq \frac{B}{1-\alpha} \cdot \frac{m_i(j)}{\OPTFk(\vec{c}_{-i})}\,,\]
as claimed.
\end{proof}

The final component needed for the proof of Lemma \ref{lemma:key} is a lower bound on the optimal fractional objective when one agent is excluded.

\begin{lemma}\label{lemma:opt-lower}
Fix any $\alpha \in (0, 1)$ and let $\vec{c}$ be a cost profile such that $v_{i^*}(k) < \frac{\alpha}{1-\alpha} \OPTFk(\vec{c}_{-i^*})$. For every agent $i \in N$ it holds that
\begin{equation*}
 \OPTFk(\vec{c}_{-i}) > \frac{1-\alpha}{\alpha}\,\big(\V(\vec{x}(\vec{c}))-m_{\ell}(x_{\ell})\big)\,.
\end{equation*}
\end{lemma}
\begin{proof}
 By the stopping condition of the \textbf{while} loop, we have:
 \begin{align*}
 \V(\vec{x}(\vec{c}))-m_{\ell}(x_{\ell}) &< \alpha \OPTFk(\vec{c})
 \leq \alpha \left(\OPTFk(\vec{c}_{-i}) + v_i(k) \right)
 \leq \alpha \OPTFk(\vec{c}_{-i})\Bigg(1 + \frac{v_{i^*}(k)}{\OPTFk(\vec{c}_{-i^*})}\Bigg) \\
 & < \alpha \OPTFk(\vec{c}_{-i})\left(1 + \frac{\alpha}{1-\alpha}\right)= \frac{\alpha}{1-\alpha}\OPTFk(\vec{c}_{-i})\,.
 \end{align*}
 The third inequality follows by the definition of $i^*$, whereas the last inequality follows directly by the assumption of the lemma.
\end{proof}

We finally present the proof of Lemma \ref{lemma:key}.

\begin{proof}[Proof of Lemma \ref{lemma:key}]
We can upper bound the total payments as follows:
\begin{align*}
 \sum_{i=1}^np_i(\vec{c}) &= \sum_{i \in W(\vec{x}(\vec{c}))}p_i(\vec{c}) \ 
 = \sum_{i \in W(\vec{x}(\vec{c}))}\sum_{j=1}^{x_i(\vec{c})} p_{ij}(\vec{c}_{-i})\\
 &\leq \frac{B}{1-\alpha} \sum_{i \in W(\vec{x}(\vec{c}))} \sum_{j=1}^{x_i(\vec{c})} \frac{m_i(j)}{\OPTFk(\vec{c}_{-i})} 
 = \frac{B}{1-\alpha} \sum_{i \in W(\vec{x}(\vec{c}))} \frac{v_i(x_i(\vec{c}))}{\OPTFk(\vec{c}_{-i})} \\
 &= \frac{B}{1-\alpha} \left( \sum_{i \in W(\vec{x}(\vec{c}))\setminus\{\ell\}} \frac{v_i(x_i(\vec{c}))}{\OPTFk(\vec{c}_{-i})} \ \ + \ \ 
 \frac{v_{\ell}(x_{\ell}(\vec{c})) - m_{\ell}(x_{\ell}(\vec{c}))}{\OPTFk(\vec{c}_{-\ell})} \ + \ \frac{m_{\ell}(x_{\ell}(\vec{c}))}{\OPTFk(\vec{c}_{-\ell})} \right) \\ 
 &< \frac{B}{1-\alpha} \left( \frac{\alpha}{1-\alpha} \cdot \frac{\sum_{i \in W(\vec{x}(\vec{c}))} v_i(x_i(\vec{c})) - m_{\ell}(x_{\ell}(\vec{c}))}{\V(x(\vec{c}))-m_{\ell}(x_{\ell}(\vec{c}))} + \frac{m_{\ell}(x_{\ell}(\vec{c}))}{\OPTFk(\vec{c}_{-\ell})} \right) \\
 &= \frac{B}{1-\alpha} \left( \frac{\alpha}{1-\alpha} + \frac{m_{\ell}(x_{\ell}(\vec{c}))}{\OPTFk(\vec{c}_{-\ell})} \right) \,.
\end{align*}
The first equality reflects the fact that $p_{ij}(\vec{c}_{-i}) = 0$ for $i \notin W(\vec{x}(\vec{c})$, whereas the second equality is due to Lemma \ref{lemma:payment-breakdown}. The first inequality follows by applying Lemma \ref{lem:critical-upper-opt} for every agent $i \in W(\vec{x}(\vec{c}))$ and every $j\in \{1\dots, x_{i}(\vec{c})\}$. Finally, the second inequality follows by applying Lemma \ref{lemma:opt-lower} to the denominators of all the terms except from ${m_{\ell}(x_{\ell}(\vec{c}))}/{\OPTFk(\vec{c}_{-\ell})}$. 
\end{proof}

\subsection{\texorpdfstring{\MultiMechanism}{\textsc{Sort-\&-Reject}} for Large Markets}
\label{subsec:large-market}

We conclude this section with a note on the performance of \MultiMechanism for instances where no single agent can significantly impact the total value attainable by the entire market. This property is frequently observed in applications such as internet marketplaces and crowdsourcing environments. In the economics and computation literature, such markets are often referred to as \emph{large} (see e.g. \cite{mehta07} for one of the first works). The notion has also been studied in the context of budget-feasible mechanism design under multiple definitions, see e.g. the work of \citet{anari18} and \citet{jalaly18}.

Given an instance $I=(N, \vec{c}, B, k, (v_i)_{i \in N})$ of the $k$-level model, let $\theta = \theta(I):= {\max_{j \in N}m_j(1)}/{\OPTIk(\vec{c})}$ be the \emph{largeness} of its underlying market. Intuitively, an instance $I$ with a small $\theta$ implies that there is no agent in $N$ whose first offered level of service is significantly valuable for the auctioneer, when comparing with the highest value the auctioneer can attain under the optimal hiring scheme (under truthful declarations).

\begin{definition}
 \label{def:large-market}
 A $k$-level instance $I=(N, \vec{c}, B, k, (v_i)_{i \in N})$ models a \emph{large market} when $\theta \to 0$.
\end{definition}

The definition above is a natural generalization to the $k$-level model of the value-based alternative large market assumption given by \citet{anari18}. The main result here is Theorem \ref{thm:large-market-bound} which implies that for large markets the approximation ratio of \MultiMechanism becomes $2.62$ while the mechanism retains all of its good properties.

\begin{theorem}
 \label{thm:large-market-bound}
 \MultiMechanism with $\alpha=\frac{2(1-\theta)}{3+\sqrt{5 + 4 \theta}}$ is a truthful, individually rational, budget-feasible $\frac{3+\sqrt{5 + 4 \theta}}{2(1-\theta)}$-approximation mechanism for instances of the $k$-level model with largeness $\theta$, and runs in time polynomial in $n$ and $k$. 
\end{theorem}

Before proving Theorem \ref{thm:large-market-bound}, we show a natural lower bound on the largeness $\theta$ of an instance $I$ for which \MultiMechanism may allocate to more than one bidders. Lemma \ref{lemma:largeness-lower-bound} will be useful to prove the budget-feasibility of \MultiMechanism for $\theta$-large instances.
 \begin{lemma}
 \label{lemma:largeness-lower-bound}
 For every cost profile $\vec{c}$ such that $v_{i^*}(k) < \frac{\alpha}{1-\alpha}\OPTFk(\vec{c}_{-i^*})$ and every bidder $\ell \in W(\vec{x}(\vec{c}))$, it holds that 
 $\theta \geq (1-\alpha) {m_{\ell}(x_{\ell})}/{OPT_F^k(\vec{c}_{-\ell})}$. 
 \end{lemma}
 \begin{proof}
 Since $\OPTFk(\vec{c}) \leq \OPTFk(\vec{c}_{-\ell}) + v_{\ell}(k) \leq \left(1 + \frac{\alpha}{1-\alpha} \right) \OPTFk(\vec{c}_{-\ell}) = \frac{1}{1-\alpha} \OPTFk(\vec{c}_{-\ell})$, we have that ,
 \begin{equation*}
 \frac{m_{\ell}(x_{\ell})}{\OPTFk(\vec{c}_{-\ell})} \leq \frac{m_{\ell}(x_{\ell})}{(1-\alpha) \OPTFk(\vec{c})} \leq \frac{m_{\ell}(x_{\ell})}{(1-\alpha) \OPTIk(\vec{c})} \leq \frac{m_{\ell}(1)}{(1-\alpha) \OPTIk(\vec{c})} \leq \frac{\theta}{1-\alpha} \,. 
 \end{equation*}
 The second inequality follows by the fact that the fractional problem is a relaxation of the integral problem, 
 the third inequality is due to the concavity of $v_{\ell}$, whereas the last inequality follows by the definition of $\theta$.
 \end{proof}

\begin{proof}[Proof of Theorem \ref{thm:large-market-bound}]
 As in Section \ref{sec:multiple}, truthfulness and individual rationality follows by Lemma \ref{lemma:monotone} and the fact that we determine the payments according to Theorem \ref{thm:archerPayments}. Furthermore, the approximation guarantee follows by Lemma \ref{lemma:guarantee} since here $\alpha = \frac{2(1-\theta)}{3+\sqrt{5 + 4 \theta}} = \frac{3-\sqrt{5 + 4 \theta}}{2} < \frac{3-\sqrt{5}}{2}$.
 
 What remains to be shown here is that \MultiMechanism is budget-feasible for this configuration. Let $\vec{c}$ be a cost profile for an instance with largeness $\theta$. If $v_{i^*}(k) \geq \frac{\alpha}{1-\alpha}\OPTFk(\vec{c}_{-i^*})$, we work exactly as in the proof of Lemma \ref{lemma:multi-bf}. When $v_{i^*}(k) < \frac{\alpha}{1-\alpha}\OPTFk(\vec{c}_{-i^*})$ we resort to Lemma \ref{lemma:key} and Lemma \ref{lemma:largeness-lower-bound} and have
\begin{equation*}
 \sum_{i \in W(\vec{x}(\vec{c}))}p_i(\vec{c}) \leq \frac{B}{1-\alpha} \left(\frac{m_{\ell}(x_{\ell}(\vec{c}))}{\OPTFk(\vec{c}_{-\ell})} + \frac{\alpha}{1-\alpha} \right) \leq \frac{B}{1-\alpha} \left(\frac{\theta}{1-\alpha} + \frac{\alpha}{1-\alpha} \right)\,,
 \end{equation*}
 where the first inequality is obtained as in the proof of Lemma \ref{lemma:multi-bf} and the second one follows from Lemma \ref{lemma:largeness-lower-bound} above.
 Budget-feasibility follows by observing that for $\alpha = \frac{3-\sqrt{5 + 4 \theta}}{2}$, it holds that $\theta + \alpha=(1-\alpha)^2$.
\end{proof}

Note that the approximation guarantee of Theorem \ref{thm:large-market-bound} improves almost linearly for small $\theta$. For instance, for $0\le \theta \le 0.04$ (i.e., in the case where no single level of service is more than 4\% of the optimal value), the approximation guarantee is roughly $2.618+3.17\theta$ (in the sense that,  for this range of $\theta$, we have $\left\lvert \frac{3+\sqrt{5+4\theta}}{2\left(1-\theta\right)} - (2.618+3.17\theta) \right\rvert<10^{-3}$). Moreover, the theorem bridges the gap between the threshold for budget-feasibility of Lemma \ref{lemma:multi-bf} (which forces the approximation ratio to be at least $2+\sqrt{3}$) and the approximation guarantee bottleneck of Lemma \ref{lemma:guarantee}.

\section{Mechanisms for Divisible Agents}\label{sec:divisible}

We consider the divisible agent model and derive two truthful, budget-feasible mechanisms with constant approximation guarantees. The first mechanism is obtained by discretizing the valuation functions and reducing the general case of the problem to the $k$-level model (Section~\ref{subsec:interpolation}). The second mechanism is an improved $2$-approximation for the special case where all agents have linear valuation functions (Section~\ref{subsec:2apx}). Note that the $e/(e-1)$ lower bound of \citet{anari18} holds even in the latter restricted setting, as the instance they construct has a linear valuation function.

\subsection{A Constant Factor Approximation Using \texorpdfstring{\MultiMechanism}{\textsc{Sort-\&-Reject}}} \label{subsec:interpolation}

Recall that in the divisible agent model, we have $\vec{x} \in [0,1]^n$ and concave non-decreasing valuation functions $\bar{v}_i: [0,1] \to \mathbb{R}_{\geq 0}$ with $\bar{v}_i(0) = 0$ for all $i \in N$. 

There is a natural correspondence between the setting with $k \ge 1$ levels of service and the setting with divisible agents: 
If we subdivide the $[0,1]$ interval into $k$ chunks of length $\frac1k$ and evaluate the $\bar{v}_i(\cdot)$'s at $\frac{1}{k}, \frac{2}{k}, \dots, \frac{k}{k}$, then this can be interpreted as the value of hiring $1, 2, \dots, k$ levels of service, respectively. 
We can then obtain a constant approximation for the setting with divisible agents by applying this discretization for $k=n$ and then using \MultiMechanism from Section \ref{sec:multiple}. Our \discretization mechanism exploits this idea; a detailed description is given in Mechanism~\ref{algo:interpolation} below.

\begin{mechanism}[h]
\caption{\discretization} 
\nonl $\hspace{-2.3ex}\rhd$ {\bf{Input:}} An instance $I=(N, \vec{\bar{c}}, B, (\bar{v}_i)_{i \in N})$ of the divisible agent model.\\
\label{algo:interpolation}
Initialize for each $i \in N$, $v_i: \{0, \dots, n\} \to \mathbb{R}_+$ with $v_i(0)=0$. \;
\For{$i \in N$ \textnormal{and} $j\in\{1,\dots, n\}$} 
{Set $v_{i}(j) = \bar{v}_{i}\left({j}/{n}  \right)$.} \label{line:cns_line_3}

\For{$i \in N$} {Set $c_{i} = {\bar{c}_i } / {n}$.} \label{line:cns_line_5}

Let $J = (N, \vec{c}, B, n, (v_i)_{i \in N})$ be the resulting discretized instance of the $k$-level model. \\ \tcp{Note that $k=n$, i.e., we construct an instance with $n$ levels of service.} \label{line:cns_line_6}
Set $\alpha ={(2+\sqrt{3})^{-1}}$ and $k=n$, and compute $\vec{x}$ by running $\MultiMechanism$ on $J$.\;
\For{$i \in N$}{Set $\bar{x}_i = {x_i} / {n}$.} \label{line:cns_line_9}
Allocate $\vec{\bar{x}}$ and determine the payments $\vec{\bar{p}}$ according to \eqref{eq:payment-id-BF}. \label{line:cns_line_10} 
\end{mechanism}

\begin{theorem}
    \label{thm:interpolation-thm}
    \discretization  is a truthful, individually rational, budget-feasible $\big(4+2\sqrt{3}\big)$-appro\-xi\-mation  mechanism for instances of the divisible agent model, and runs in time polynomial in $n$.
\end{theorem}

For notational convenience, throughout Section \ref{subsec:interpolation}, we use $\vec{\bar{c}}$ to denote cost profiles in the divisible agent model and $\vec{c}$ to denote cost profiles in the $k$-level model. The following lemma will be key in proving the approximation guarantee of the theorem.

\begin{lemma}
\label{lemma:OPTF-to-OPTFk}
Let $I=(N, \vec{\bar{c}}, B, (\bar{v}_i)_{i \in N})$ be an instance of the divisible agent model. Consider the instance $J = (N, \vec{c}, B, n,\allowbreak (v_i)_{i \in N})$ as defined in line \ref{line:cns_line_6} of the mechanism \discretization. It holds that $\OPTF(\vec{\bar{c}}) \leq 2 \cdot \OPTFn(\vec{c})$.
\end{lemma}

\begin{proof}
Let $\vec{x}^* = (x^*_1,\ldots, x^*_n) \in [0,1]^n$ be an optimal solution to the respective Fractional Concave Knapsack Problem for $I$ (see \eqref{eq:fKP} in Section \ref{subsec:divisible-def}). Let $z_i := \floor{x^*_i n}$ for all $i \in N$, and  $\vec{z} := (z_1, \dots, z_n)$. 
Define $\vec{1} = (1, 1, \ldots, 1)$. 
We have:

\begin{align*}
     \OPTF(\vec{\bar{c}}) = \sum_{i \in N} \bar{v}_i(x^*_i) 
     &\leq \sum_{i \in N} \bar{v}_i \bigg( \min \bigg\{\frac{z_i+1}{n}, 1 \bigg\} \bigg)
     = \sum_{i \in N} v_i\left(\min \left\{z_i+1, n \right\}\right) \\
     &= \sum_{i \in N}   v_i\left(\min \left\{z_i+1, n \right\}-1\right) + m_i\left(\min \left\{z_i+1, n \right\}\right) \\
     &\leq \sum_{i \in N} v_i(z_i) + v_i(1)   = v(\vec{z}) + v(\vec{1})\,.\numberthis \label{eq:ub-two-feasible-sols}
\end{align*}
For the first inequality, we use the facts that $x^*_i \leq 1$ and $x^*_i n \leq \floor{x^*_i n} + 1=z_i+1$ for all $i \in N$. The second equality follows by expressing $\bar{v}_i$ in terms of $v_i$ for each $i \in N$ (see line \ref{line:cns_line_3} of \discretization). The final inequality is due to the monotonicity of the functions $\bar{v}_i$ and $v_i$ for each $i \in N$, as well as their concavity,  which implies that $m_i(j) \leq m_i(1) \leq v_i(1)$ for any $j \in [n]$. 

Next, we show that both $\vec{z}$ and $\vec{1}$ are feasible solutions for $J$. For $\vec{z}$, it holds by construction that
\begin{equation*}
    \sum_{i \in N} c_i z_i = \sum_{i \in N} \frac{\bar{c}_i}{n} z_i = \sum_{i \in N} \frac{\bar{c}_i}{n} \floor{x^*_i n} \leq \sum_{i \in N} \bar{c}_i x^*_i \leq B \,.
\end{equation*}
The last inequality follows from the fact that $\vec{x}^*$ is a feasible solution for $I$ as it is, by assumption, optimal.

Similarly, we can verify that $\vec{1}$ is a feasible solution for $J$ by observing that
\begin{equation*}
    \sum_{i \in N} c_i = \sum_{i \in N} \frac{\bar{c}_i}{n} = \frac{\sum_{i \in N} \bar{c}_i}{n} \leq \frac{n \cdot \max_{i \in N} \bar{c}_i}{n} = \max_{i \in N} \bar{c}_i \leq B\,,
\end{equation*}
where  
the final inequality holds by the cost restriction assumption of the divisible agent model (see Section \ref{subsec:divisible-def}).

Since both $\vec{z}$ and $\vec{1}$ are feasible solutions for $J$, we obtain 
\[
v(\vec{z}) \leq \OPTFn(\vec{c})
\quad\text{and}\quad
v(\vec{1}) \leq \OPTFn(\vec{c})\,.
\]
By combining this with 
\eqref{eq:ub-two-feasible-sols},
the lemma follows.
\end{proof}

\begin{proof}[Proof of Theorem~\ref{thm:interpolation-thm}]

\discretization reduces the divisible instance to an instance with $k=n$ multiple levels of service and uses \MultiMechanism parameterized with $\alpha=(2+\sqrt{3})^{-1}$ to compute an allocation. Thus, the mechanism inherits the truthfulness and individual rationality properties through Lemma~\ref{lemma:monotone}. Moreover, for the approximation guarantee we have
\[
    \bar{v}(\vec{\bar{x}}(\vec{\bar{c}})) = \V(\vec{x}(\vec{c})) \geq \alpha\OPTFn(\vec{c}) \geq \frac{\alpha}{2}\OPTF(\vec{\bar{c}})\,,
\]
where the first inequality follows from Lemma \ref{lemma:guarantee} and the second inequality follows from Lemma \ref{lemma:OPTF-to-OPTFk}.

\sloppy
What remains to be shown is that \discretization is budget-feasible. We will show that, for an instance $I=(N, \vec{\bar{c}}, B, (\bar{v}_i)_{i \in N})$ of the divisible agent model, the payments prescribed by line \ref{line:cns_line_10} of \discretization coincide with the payments made by \MultiMechanism for the instance $J=(N, \vec{c}, B, n, (v_i)_{i \in N})$ which is  constructed in line \ref{line:cns_line_6} of \discretization. Specifically, we will show that for each agent $i \in N$, it holds that $\bar{p}_i(\vec{\bar{c}})=p_i(\vec{c})$. Let $w=f(x)={x}/{n}$. Indeed, for each $i \in N$, we have:
\begin{align*}
    \bar{p}_i(\vec{\bar{c}}) &= \bar{c}_i \cdot \bar{x}_i(\vec{\bar{c}}) + \int_{\bar{c}_i}^B \bar{x}_i(z, \vec{\bar{c}}_{-i})\,dz 
    = c_i n \cdot \frac{x_i(\vec{c})}{n} + \int_{c_i n}^B \frac{x_i\left({z}/{n}, \vec{c}_{-i}\right)}{n}\,dz \\
    &= c_i x_i(\vec{c}) + \int_{c_i n}^B \frac{1}{n} \cdot x_i\left(\frac{z}{n}, \vec{c}_{-i}\right)\,dz 
    = c_i x_i(\vec{c}) + \int_{c_i n}^B f'(z) \cdot x_i\left(f(z), \vec{c}_{-i}\right)\,dz \\
    &= c_i x_i(\vec{c}) + \int_{f(c_i n)}^{f(B)} x_i\left(w, \vec{c}_{-i}\right)\,dw 
    = c_i x_i(\vec{c}) + \int_{c_i}^{{B}/{n}} \!\!\! x_i\left(w, \vec{c}_{-i}\right)\,dw = p_i(\vec{c}).
\end{align*}
The first equality follows from applying Theorem \ref{thm:archerPayments} (Myerson's payment identity) to the allocation $\bar{\vec{x}}(\bar{\vec{c}})$ for $I$ in \discretization. The second equality follows from lines \ref{line:cns_line_5} and \ref{line:cns_line_9}, since, by definition, $\bar{x}_i(\vec{\bar{b}}) = {x_i\left(\frac{1}{n}\vec{\bar{b}}\right)}/{n}$ holds for all $\vec{\bar{b}} \in [0,B]^{n}$; here $\vec{\bar{b}} = (z, \vec{\bar{c}_{-i}})$ and  $\frac{1}{n}(z, \vec{\bar{c}_{-i}}) = (z/n, \vec{{c}_{-i}})$. The subsequent four equalities are obtained through simple calculus, involving a change of the variable of the integral term. Finally, the last equality holds because \MultiMechanism also prescribes payments according to Theorem \ref{thm:archerPayments} for instance $J$. Thus, as \MultiMechanism is budget-feasible by Lemma \ref{lemma:multi-bf}, budget-feasibility of \discretization follows.

\fussy
We conclude the proof of the theorem by noting that \discretization runs in time polynomial in $n$. Recall that, by Theorem \ref{thm:multi-mechanism}, Mechanism \MultiMechanism runs in time polynomial in $n$ and $k$. Since Mechanism \discretization, given an instance $I$ with $n$ agents, constructs an instance of the $n$-level model, i.e.,  $k=n$, its running time is polynomial in $n$.
\end{proof}

\color{black}

\subsection{An Improved Mechanism for Linear Valuation Functions}
\label{subsec:2apx}

The best known approximation guarantee for the special case where all divisible agents have \emph{linear} valuation functions is $(3 + \sqrt{5}) / 2 \approx 2.62$ by \citet{klumper22}. 
Below, we improve upon this by presenting a simple budget-feasible $2$-approximation mechanism for this setting, thereby establishing that the optimal approximation ratio for this problem lies in the interval $[e/(e-1), 2]$. 
Our mechanism is inspired by the randomized $2$-approximation mechanism by \citet{gravin20} for indivisible agents. To avoid unnecessarily heavy notation in this special case, let $v_i := \bar{v}_i(1)$ and, thus, $\bar{v}_i(x) = v_i x$, for all $i\in N$.

\subsubsection{Phase 1: Pruning Mechanism for Divisible Agents}
\label{subseq:pruning}

We first extend the \emph{\pruning} mechanism of Gravin et al.~\cite{gravin20} to the divisible setting. This mechanism constitutes a crucial building block for both their deterministic $3$-appro\-xi\-mation mechanism and their randomized $2$-approximation mechanism for indivisible agents \cite{gravin20}. As we show below, it serves as a useful starting point for the divisible setting as well. 

Given a profile $\Vec{c}$, this mechanism computes an allocation $\bar{\Vec{x}}(\Vec{c})$, which we refer to as the \emph{provisional allocation}, and a positive quantity $r(\Vec{c})$, which we refer to as the \emph{rate}. We assume that the agents are initially relabeled by their decreasing whole-value-per-cost ratio, i.e., $\frac{v_1}{c_1} \geq \frac{v_2}{c_2} \geq \dots \geq \frac{v_n}{c_n}$. The mechanism proceeds as described in Mechanism~\ref{algo:pruning}.

\begin{mechanism}[ht]
\caption{\pruning\ \ {\small (by Gravin et al.~\cite{gravin20})}}
\label{algo:pruning}
\nonl $\hspace{-2.3ex}\rhd$ {\bf{Input:}} $I=(N, \vec{\bar{c}}, B, (v_i)_{i \in N})$ with the agents relabeled so that $\frac{v_1}{c_1} \geq \dots \geq \frac{v_n}{c_n}$\\
Let $r=\max\{v_i \mid i \in N\} / B$.\\
\For{$i \in N$}{
 Set $\bar{x}_i = 1$ if $\frac{v_i}{c_i} \geq r$ and $\bar{x}_i = 0$ otherwise.
 } 
Let $\ell = \argmax\{i \mid \bar{x}_i=1\}$. \;
\While {$rB < \sum_{i=1}^{\ell}v_i - \max\{v_i \mid 1\le i \le \ell\}$}
{Continuously increase rate $r$.\\
\lIf{$\frac{v_{\ell}}{c_{\ell}}\leq r$}{set $\bar{x}_{\ell}=0$ and $\ell=\ell-1$.} 
}
\Return $(r,\bar{\vec{x}})$\\
\end{mechanism}

\citet{gravin20} showed that \pruning\ is monotone. In fact, an even stronger \emph{robustness} property holds (and is implicit in the proof of Lemma 3.1 in \cite{gravin20}): 
each agent $i$ that is a winner in the provisional allocation cannot alter the outcome of \pruning\ unilaterally while remaining a winner in the provisional allocation.

\begin{lemma}[implied by Lemma 3.1 of \cite{gravin20}]
 \label{lemma:pruning-invariant}
 Let $\vec{c}$ be a cost profile. Consider an agent $i \in N$ with $\bar{x}_i(\vec{c})=1$. Then, for all $c_i'$ such that $\bar{x}_i(c_i', \vec{c}_{-i})=1$, it holds that $\bar{\vec{x}}(c_i', \vec{c}_{-i})=\bar{\vec{x}}(\vec{c})$ and $r(c_i', \vec{c}_{-i}) = r(\vec{c})$.
\end{lemma}

Given this robustness property, \pruning\ can be used as a first filtering step to discard inefficient agents, followed by a subsequent allocation scheme which takes $(r(\vec{c}), \bar{\vec{x}}(\vec{c}))$ as input. If the subsequent allocation scheme is monotone, then the sequential composition of \pruning\ with this allocation scheme is monotone as well. This \emph{composability property} is proven in Lemma 3.1 of \cite{gravin20}.

Let $(r, \bar{\vec{x}})$ be the output of \pruning\ for a cost profile $\vec{c}$. Given $\bar{\vec{x}}$, we define $S$ as the set of agents that are provisionally allocated, $i^*$ as the highest value agent in $S$ (where ties are broken lexicographically), and $T$ as the set of remaining agents. Formally, 
\begin{equation}
\label{eq:set-def-pruning}
S = \sset{i \in N}{\bar{x}_i = 1}, \quad 
i^* \in\argmax\sset{v_i}{i \in S}, \quad \text{and} \quad
T = S \setminus \set{i^*}\,.
\end{equation}
Note that the definitions of $S$, $i^*$ and $T$ depend on $\bar{\vec{x}}$ (and thus the cost profile $\vec{c}$). For notational convenience, we do not state this reference explicitly if it is clear from the context. Similarly, here we typically drop the argument of $\OPTF(\vec{c})$. 

The following properties were proved in \cite{gravin20} and are useful in our analysis. 
\begin{lemma}[Lemma 3.2 of \cite{gravin20}]
 \label{lemma:pruning-properties}
 Given a profile $\vec{c}$, let $(r, \bar{\Vec{x}})$ be the output of\, \pruning. Let $S = T \cup \set{i^*}$ be defined as in \eqref{eq:set-def-pruning} with respect to $\bar{\vec{x}}$. We have 
 \begin{enumerate}[labelindent=\parindent,leftmargin=*,itemsep=3pt,topsep=3pt]
 \item[(i)] $c_i \leq \frac{v_i}{r} \leq B$ for all $i \in S$. 
 \item[(ii)] $\V(T) \leq rB < \V(S)$.
 \item[(iii)] $\OPTF \leq \V(S)+ r\cdot (B - c(S))$. 
 \end{enumerate}
\end{lemma}

\subsubsection{Phase 2: The Final Allocation Scheme} 

Our mechanism combines the \pruning\ mechanism above with the allocation scheme defined in \eqref{eq:ind-alloc} below. We refer to the resulting mechanism as \emph{\paam} (see Mechanism~\ref{mechanism:posted-price}).

First, we need to define the following constants: 
\[
q_{i^*} = 
\begin{cases} 
\frac{1}{2} - q & \!\text{if $v_{i^*} \le \V(T)$} \\
\frac{1}{2} & \!\text{otherwise}
\end{cases},
\qquad
q_i = 1 - q_{i^*} - q \quad \forall i \in T, 
\quad\text{where}\quad
q = \frac{1}{2} \frac{\V(S) - rB}{\min\{v_{i^*}, \V(T)\}}\,.
\]
Note that the constant $q_i$ for all agents $i \in T$ is the same. It is not hard to prove that $q \in \big[0, \frac12\big]$ (see \cite[Lemma 5.1]{gravin20}). The constants above are chosen so that
\begin{equation}
 \label{eq:identity}
 \frac{rB}{2} = q_{i^*} v_{i^*} + (1-q_{i^*}-q) \V(T)\,.
\end{equation}

For given $r$ and $S$, as is the case in line \ref{line:paam-assign} of \paam, we are going to define our (fractional) allocation function $x_i(\vec{c})$, for $i \in S$, so that it only depends on agent $i$'s cost $c_i$ and, thus, we will slightly abuse notation and write $x_i(c_i)$ in what follows. For each agent $i \in S$ let 
\begin{align}
x_i(z)
&= 
q_{i} + \frac{v_{i} - r z}{2v_{i}} \quad \text{for} \ \ z \in \Big[0, \frac{v_i}{r}\Big]\,.
\label{eq:ind-alloc}
\end{align}
It is not hard to verify that $0\le x_i(z)\ \le 1$, given the chosen parameters $q_{i^*}$, $q_T$ and $q$ above. 
Further, by property (i) of Lemma~\ref{lemma:pruning-properties}, the cost $c_i$ of each agent $i \in S$ is at most ${v_i}/{r}$ and, therefore, $x_i(c_i)$ is well-defined.

\begin{mechanism}[ht]
\caption{\paam\ }
\label{mechanism:posted-price}
\nonl $\hspace{-2.3ex}\rhd$ {\bf{Input:}} $I=(N, \vec{\bar{c}}, B, (v_i)_{i \in N})$  with the agents relabeled so that $\frac{v_1}{c_1} \geq \frac{v_2}{c_2} \geq \dots \geq \frac{v_n}{c_n}$ \;
Obtain $(r, \bar{\vec{x}})$ by running \pruning\ for profile $\vec{c}$. \;
Let $S = \sset{i \in N}{\bar{x}_i = 1}$,\: $i^* \in\argmax\sset{v_i}{i \in S}$,\: $T = S \setminus \set{i^*}$. \;
Determine the fractional allocations $x_i(c_i)\,, \forall i \in S$, as in \eqref{eq:ind-alloc}, and set $x_i(c_i) = 0\,, \forall i \in N\setminus S$. \label{line:paam-assign}\;
Allocate $\vec{x}$ and determine the payments $\vec{p}$ according to \eqref{eq:payment-id-BF}. \;
\end{mechanism}

\begin{theorem}
 \label{thm:prune-and-assign}
 \paam\ is a truthful, individually rational, budget-feasible $2$-approximation mechanism for instances of the divisible agent model with linear valuation functions, and runs in time polynomial in $n$.
\end{theorem}

Given the lower bound of $1+\sqrt{2} \approx 2.41$ by \citet{chen11} in this setting with indivisible agents, Theorem \ref{thm:prune-and-assign} establishes a separation between the indivisible agent model and the divisible agent model with linear valuation functions.
The theorem follows by the three lemmata below.

\begin{lemma}
\paam\ is monotone.
\end{lemma}
\begin{proof}
 As argued above, \pruning\ is monotone and it suffices to show that the allocation scheme in \eqref{eq:ind-alloc} is monotone. 
 Fix an arbitrary agent $i$ and a cost profile $\vec{c}$. 
 Let $c'_i < c_i$. 
 We need to prove that $x_i(c'_i, \vec{c}_{-i}) \ge x_i(\vec{c})$.
 
 Let $\bar{\vec{x}}(\vec{c})$ be the provisional allocation obtained from \pruning. If $\bar{x}_i(\vec{c}) = 0$ the claim follows trivially. 
 Assume $\bar{x}_i(\vec{c}) = 1$. Because \pruning\ is monotone, we have $\bar{x}_i(c'_i, \vec{c}_{-i}) = 1$. Also, because \pruning\ is robust the output remains the same, i.e., $(\bar{\vec{x}}(\vec{c}), r(\vec{c})) = (\bar{\vec{x}}(c'_i, \vec{c}_{-i}), r(c'_i, \vec{c}_{-i}))$. In particular, this implies that the respective rates that determine the fractional allocation in \eqref{eq:ind-alloc} are the same.
 The claim now follows by observing that the allocation functions in \eqref{eq:ind-alloc} are monotone non-increasing in $z$.
\end{proof}

\begin{lemma}
\paam\ has an approximation ratio of \, $2$.
\end{lemma}
\begin{proof}
 By property (iii) of Lemma~\ref{lemma:pruning-properties}, we have 
 \begin{align*}
 \OPTF
 & \leq \V(S)+ r \cdot (B - c(S) )
 = \sum_{i \in S} (v_{i} - r c_{i}) + rB\,.
 \end{align*}
 By the definition in \eqref{eq:ind-alloc}, we have for every $i \in T \cup \set{i^*}$:
 \[
 x_i(z) = q_i + \frac{v_i - r z}{2v_i}, 
 \quad \text{or, equivalently,} \quad
 v_i - rz = 2v_i (x_i(z) - q_i)\,.
 \]
 Combining these two inequalities above, we obtain 
 \begin{align*}
 \OPTF 
 & \le 2 \sum_{i \in S} v_i x_i(c_i) - 2 \sum_{i \in S} v_i q_i + rB 
 = 2 \V(\vec{x}) - 2 (q_{i^*} v_{i^*} + (1-q_{i^*}-q) \V(T)) + rB = 2 \V(\vec{x}(\vec{c}))\,, 
 \end{align*}
 where for the last equality we exploit the identity in \eqref{eq:identity}. This concludes the proof.
\end{proof}

\begin{lemma}
\paam\ is budget-feasible.
\end{lemma}
\begin{proof}
 Given a cost profile $\vec{c}$, the payment of agent $i\in S$ is defined as
 \[
 p_i(\vec{c}) 
 = c_i x_i(\vec{c}) + \int_{c_i}^{B}\! x_i(z, \vec{c}_{-i}) dz
 = c_i x_i(c_i) + \int_{c_i}^{{v_i}/{r}} \!\!\!x_i(z) dz
 \le \int_{0}^{{v_i}/{r}} \!\!\!x_i(z) dz \,.
 \]
 Here the second equality holds because \pruning\ ensures that $x_i(z,\vec{c}_{-i}) = 0$ for $z > v_i/r$. The inequality holds because $x_i(\cdot)$ is strictly decreasing in $[0, c_i]$. 
 
 Consider an agent $i \in T \cup \set{i^*}$. We have 
 \begin{align*}
 \int_{0}^{v_i/r} \! \! \! x_i(z) dz 
 & = \int_{0}^{{v_{i}}/{r}} \! \left(q_i + \frac{v_{i} - r z}{2v_{i}}\right) dz 
 = \frac{v_{i}}{r} \left( q_i + \frac{1}{2}\right) - \frac{r}{2v_{i}} \int_{0}^{{v_{i}}/{r}} \!\!\! z dz 
 = \frac{v_{i}}{r} \left( q_i + \frac{1}{4} \right)\,. 
 \end{align*}
 Summing over all agents, we obtain that the total payment is at most
 \begin{align*}
 \sum_{i=1}^np_i(\vec{c})=\sum_{i \in S} p_i(\vec{c}) 
 & \le \sum_{i \in S} \frac{v_{i}}{r} \left( q_i + \frac{1}{4} \right)
 = \frac{1}{r} \left(q_{i^*} v_{i^*} + (1-q_{i^*}-q) \V(T) \right) + \frac{\V(S)}{4r}
 = \frac{B}{2} + \frac{\V(S)}{4r}\,,
 \end{align*}
 where the last equality uses the identity in \eqref{eq:identity}. 
 
 We conclude the proof by showing that $\V(S)/4r \le B/2$, or, equivalently, $rB \ge \V(S)/2$. Starting with \eqref{eq:identity} and applying the definition of $q_{i^*}$, we have
 
 \begin{align*}
 \frac{rB}{2} 
 & = (q_{i^*} v_{i^*} + (1-q_{i^*}-q) \V(T)) 
 = \frac12 ( v_{i^*} + \V(T)) - q \cdot \min\set{v_{i^*}, \V(T)} 
 \ge \frac12 \V(S) - q \V(T) \,.
 \end{align*} 
 
 The claim follows by observing that $\V(T) \le rB$ by property (ii) of Lemma~\ref{lemma:pruning-properties} and $q \le \frac12$ by definition.
\end{proof}

We conclude the section by showing that our analysis of \paam is tight. Indeed, consider an instance with $2$ agents with $v_1=v_2=1$, $c_1=\epsilon, c_2=1-\epsilon$, with $\epsilon \in (0,1)$, and the budget of the auctioneer being $B=1$. The optimal hiring scheme is $\vec{x}^*=(1,1)$ and, therefore, $\V(\vec{x}^*(\vec{c}))=2$. Consider now the outcome of \paam for this instance. Initially, \pruning does not discard either bidder since the \textbf{while} condition for $r=1/B$ evaluates to \emph{False}. Thus, $S=\{1,2\}$ and $r=1/B$. Then, it is easy to observe the final hiring scheme of \paam prescribes an allocation of $x_1(0)=x_2(1)=1/2$ (by viewing bidder $1$ as $i^*$). Therefore $\V(\vec{x}(\vec{c}))=1$, which matches the guarantee of Theorem \ref{thm:prune-and-assign}.

\section{Conclusion and Future Work}

In this work we revisited two budget-feasible mechanism design settings where partial allocations are
allowed and draw clear connections between them. Under the mild assumption of being able to afford each agent entirely, we give deterministic, truthful and budget-feasible mechanisms with constant approximation guarantees. We believe these are settings that are both interesting and relevant to applications and there are several open questions we do not settle here. A natural direction, not considered at all in this work, is to deal with additional combinatorial constraints, like matching, matroid, or even polymatroid (for the $k$-level setting) constraints. For the 
$k$-level setting, it would also be interesting to understand whether we can obtain mechanisms with approximation guarantees closer to those possible for single-level settings, or alternatively, determine whether allowing multiple levels of service is an inherently harder problem. At the same time, it would be interesting to obtain an improved approximation guarantee for concave valuation functions in the divisible agents model. A possible avenue here could be to come up with a more nuanced discretization procedure than that of \discretization. Finally, as far as simple settings are concerned, the most important open problem is still the indivisible agents case with additive valuation functions, for which the best-possible approximation ratio is in $[1+\sqrt{2},\ 3]$ (due to \cite{chen11, gravin20}). The corresponding range for the divisible agent setting is $[e/(e-1),\ 2]$ (due to \cite{anari18} and our Theorem \ref{thm:prune-and-assign}). Any progress on these fronts may give rise to novel techniques, which could further be used for problems in richer environments.

\section*{Acknowledgements}
This work was supported by the Dutch Research Council (NWO) through its Open Technology Program, proj.~no.~18938, and the Gravitation Project NETWORKS, grant no.~024.002.003.
It has also been partially supported by project MIS 5154714 of the National Recovery and Resilience Plan Greece 2.0, funded by the European Union under the NextGenerationEU Program and the EU Horizon 2020 Research and Innovation Program under the Marie Skłodowska-Curie Grant Agreement, grant no.~101034253.
Moreover, it was partially supported by the project MIS 5154714 of the National Recovery and Resilience Plan Greece $2.0$ funded by the European Union under the NextGenerationEU Program and by the Hellenic Foundation for Research and Innovation (HFRI), partially by the ``1st Call for HFRI Research Projects to support faculty members and researchers and the procurement of high-cost research equipment'' (proj.~no.~HFRI-FM17-3512) and partially by the HFRI call “Basic research Financing (Horizontal support of all Sciences)” under the National Recovery and Resilience Plan “Greece 2.0” funded by the European Union – NextGenerationEU (proj.~no.~ 15877). Finally, part of this work was done during AT's visit to the University of Essex that was supported by COST Action CA16228 (European Network for Game Theory).

\bibliographystyle{ACM-Reference-Format}
\bibliography{arxiv_final/references_final}

\appendix

\section{\texorpdfstring{\BestInMechanism}{\textsc{Greedy-Best-In}}: Almost tight for the Best-in Setting}
\label{app:unrestricted_k_levels}
As described in Section \ref{subsec:multi-def}, throughout this work we consider the $k$-level model under the \emph{all-in setting}, i.e., given a cost profile $\vec{c}$, for every agent $i \in N$ it holds that $k \cdot c_i \leq B$. However, the $k$-level model has also been studied in the \emph{best-in setting} (see \citep{chan14}), 
where the weaker assumption that $c_i \leq B$ for all $i \in N$ is assumed instead. In this appendix, we show that the truthful and budget-feasible mechanism \MultiMechanism (Mechanism \ref{mechanism:multi} defined in Section~\ref{sec:BF-Greedy}), 
requires just a few modifications to achieve an approximation guarantee that is almost tight for this setting. More precisely, for $k$ levels of service and concave valuation functions, we re-parameterize our mechanism \MultiMechanism and obtain a guarantee of $k+2+ o(1)$, almost matching the known lower bound of $k$ of \citet{chan14} 
(see also Remark \ref{remark:chan_lb} in Section \ref{subsec:multi-def}). 
Note that Theorem \ref{thm:bestin} directly shows that the lower bounds of \citet{chan14} cannot be generalized for instances with multiple agents and remain linear or logarithmic to the total number of levels of service across all agents. In that sense, their upper bound is close to being tight for a constant number of agents, but its gap with the lower bound grows as an unbounded function of $n$. Instead, the guarantee of our Theorem \ref{thm:bestin} is almost best-possible \emph{for any} $n$.

For presentation purposes, we restate this parameterized version of the mechanism separately in this section, ``rebrand'' it as \BestInMechanism and analyze it in a mostly self-contained manner. 

\begin{mechanism}[ht]
\caption{\BestInMechanism\label{mech:bestin}}
\nonl $\hspace{-2.3ex}\rhd$ {\bf{Input:}} A profile $\vec{c}$. \;
\label{mechanism:bestin}
Set $\alpha= \big(3+k-\sqrt{9+2k+k^{2}\,}\big) \big/ 2k$ and $\beta= (1-2\alpha)/(\alpha k+1)$. \\

Set $i^* \in \argmax_{i \in N} v_i(1)/\OPTFk(\vec{c}_{-i})$. \tcp*{we break ties lexicographicaly with respect to the agents' indices}

\If{$v_{i^*}(1) \geq \beta \cdot \OPTFk(\vec{c}_{-i^*})$ \label{line:best-in-beta}}{Set $x_{i^*}=1$ and $x_i=0$ for all $i \in N \setminus \{i^*\}$.}
\Else{
Compute an optimal fractional allocation $\vec{x}^*\!(\vec{c})$ using Algorithm \ref{algo:fk-knapsack}. \tcp*{i.e., $\vec{x}^*\!(\vec{c})$ is almost integral}

Initialize $\vec{x}= (\floor{x^*_1(\vec{c})}, \dots, \floor{x^*_n(\vec{c})} )$. \label{line:init_rounding} \tcp*{by the construction of $\vec{x}^*\!(\vec{c})$ and the best-in assumption, $\vec{x}\neq \vec{0}$}

\For(\tcp*[f]{recall that $W(\vec{x})=\{i \in N \mid x_{i} > 0\}$}){$i \in W(\vec{x})$ }{
\For{$j \in \{1, \dots, x_i\}$}{
Add the marginal value-per-cost ratio $m_i(j)/c_i$ to a list $\mathcal{L}$.
}
}

Sort $\mathcal{L}$ in non-increasing order and let $\ell$ be the index of the last agent of $W(\vec{x})$ in $\mathcal{L}$.

\While{$\V(\vec{x}) -m_{\ell}(x_{\ell}) \geq \alpha \OPTFk(\vec{c})$ \label{line:BIM_while}} 
{

Set $x_{\ell}= x_{\ell}-1$. \;

Remove the last element from $\mathcal{L}$ and update $\ell$. \;
}}
Allocate $\vec{x}$ and determine the payments $\vec{p}$ according to \eqref{eq:payment-id-BF}.
\end{mechanism}

\begin{theorem}
 \label{thm:bestin}
 Mechanism \BestInMechanism is a truthful, individually rational, budget-feasible $(k+2+o(1))$-approximation mechanism for instances of the $k$-level model and the best-in setting, and runs in time polynomial in $n$ and $k$.
\end{theorem}

Intuitively, the fact that the mechanism may allocate a single level of service of an agent $i^*$ while excluding everyone else---even if more levels of $i^*$ could be afforded at a first glance---makes the \textbf{if} part of \BestInMechanism significantly weaker than that of \MultiMechanism. However, this is needed to ensure budget-feasibility.

To prove Theorem \ref{thm:bestin}, similar arguments to those of Section \ref{sec:BF-Greedy} are used, albeit slightly adapted for the best-in environment.
We first present two auxiliary lemmata, one lemma for the proof of the approximation guarantee and one lemma which will prove to be useful for showing budget-feasibility for the best-in setting.
\begin{lemma}
 \label{lemma:best-in-guarantee}
 For the allocation $\vec{x}(\vec{c})$ computed by \BestInMechanism with input the cost profile $\vec{c}$, it holds that $\V(\vec{x}(\vec{c}) \geq \alpha \OPTIk(\vec{c})$, with $\alpha=\big(3+k-\sqrt{9+2k+k^{2}\,}\big) \big/ 2k$.
\end{lemma}
\begin{proof}
 
 Similarly to Lemma \ref{lemma:guarantee}, for a cost profile $\vec{c}$ we will prove the claimed guarantee against the optimal solution to the fractional relaxation of the bounded knapsack instance, i.e., we will show that $\V(\vec{x}(\vec{c})) \geq \alpha \OPTFk(\vec{c})$. We distinguish the following two cases:

\smallskip

\noindent\underline{Case 1:} $v_{i^*}(1) \geq \beta \cdot \OPTFk(\vec{c}_{-i^*})$. We directly have
\[
 v_{i^*}(1) \geq \beta \cdot \OPTFk(\vec{c}_{-i^*}) \geq \beta \left(\OPTFk(\vec{c})-v_{i^*}(k) \right) \geq \beta \left(\OPTFk(\vec{c})-kv_{i^*}(1) \right),
\]
where the last inequality follows by the concavity of $v_{i^*}(\cdot)$. Rearranging terms leads to \[v_{i^*}(1) \geq \frac{\beta}{1+\beta k} \OPTFk(\vec{c})=\alpha\OPTFk(\vec{c}),\]
which concludes the case since $v_{i^*}(1) = \V(\vec{x}(\vec{c}))$.

\smallskip 

\noindent\underline{Case 2:} 
$v_{i^*}(1) < \beta \cdot \OPTFk(\vec{c}_{-i^*})$. In this case, we immediately obtain $\V(\vec{x}(\vec{c})) \geq \alpha \OPTFk(\vec{c})$ whenever the \textbf{while} loop runs at least once. We argue that $\V(\vec{x}(\vec{c})) \geq \alpha \OPTFk(\vec{c})$ still holds for our choice of $\alpha$ and $\beta$ when the \textbf{while} condition on line \ref{line:BIM_while} always evaluates to \emph{False}. Toward a contradiction, suppose that this is not the case. For the best-in setting, it is true that, similarly to the proof of Lemma \ref{lemma:guarantee}, for the fractionally included agent $f$ it holds that
\begin{equation}\label{ineq:m_f_upper_best_in}
\frac{m_f(\ceil{x^*_f})}{\OPTFk(\vec{c})} \leq \frac{m_f(1)}{\OPTFk(\vec{c})} \leq\frac{v_f(1)}{\OPTFk(\vec{c}_{-f})} \leq \frac{v_{i^*}(1)}{\OPTFk(\vec{c}_{-i^*})} < \beta\,,
\end{equation}
where the first inequality follows from the concavity of each valuation function, the third inequality holds by the choice of $i^*$, and the last one is the very definition of Case 2 for this setting.

At the same time, \eqref{ineq:m_f_lower} holds regardless of whether we are in the best-in or the all-in setting.

By combining inequalities \eqref{ineq:m_f_lower} and \eqref{ineq:m_f_upper_best_in} we obtain that $\alpha$ must be such that $\beta > 1-\alpha$. However, it is a matter of simple calculations to verify that this is not possible for any $k > 0$ for our choice of $\alpha$ and $\beta$ leading to a contradiction.
\end{proof}

\begin{lemma}
 \label{lemma:best-in-opt-minus-i}
 Let $\vec{c}$ be a cost profile so that $v_{i^*}(1) < \beta \OPTFk(\vec{c}_{-i^*})$ in line \ref{line:best-in-beta} of \BestInMechanism. Then, for every $i \in W(\vec{x}(\vec{c}))$ it holds that
\begin{equation*}
 \OPTFk(\vec{c}) \leq \big(1+\beta k\big)\OPTFk(\vec{c}_{-i}).
 \end{equation*}
 \begin{proof}
 Indeed, by concavity of $v_i(\cdot)$, the definition of $i^*$ in \BestInMechanism and the fact that $v_{i^*}(1) < \beta \OPTFk(\vec{c}_{-i^*})$, we obtain
 \begin{align*}
 \OPTFk(\vec{c}) &\leq \OPTFk(\vec{c}_{-i}) + v_i(k) \leq \OPTFk(\vec{c}_{-i}) + kv_i(1) \\
 &\leq \OPTFk(\vec{c}_{-i}) \Bigg( 1+\frac{kv_{i^*}(1)}{\OPTFk(\vec{c}_{-{i^*}})}\Bigg) \leq \big(1+\beta k\big)\OPTFk(\vec{c}_{-i}).\qedhere
 \end{align*}
 \end{proof}
 
\end{lemma}

\noindent We are now ready to prove Theorem \ref{thm:bestin}.

\begin{proof}[Proof of Theorem \ref{thm:bestin}]
 Observe that truthfulness and individual rationality of \BestInMechanism follow by the monotonicity of the allocation rule, as proven by Lemma \ref{lemma:monotone}, and the fact that we determine the payments according to \eqref{eq:payment-id-BF}. To prove that \BestInMechanism achieves the claimed guarantee, it is sufficient to observe that, the guarantee of $1/\alpha$ we proved in Lemma \ref{lemma:best-in-guarantee} is upper-bounded by $k+2+o(k)$. Indeed, $\lim_{k \to \infty}1/\alpha=k+2$. 
 
 What remains to be proved is the budget-feasibility of \BestInMechanism. To do that, we distinguish the following cases:\smallskip

\noindent\underline{Case 1:} $v_{i^*}(1) \geq \beta \OPTFk(\vec{c}_{-i^*})$. In this case agent $i^*$ is the only agent hired and for $1$ level of service. By Theorem \ref{thm:archerPayments}, we obtain
 $p_{i^*}(\vec{c}) = p_{i^*1}(\vec{c}_{-i^*}) \leq B$ and budget-feasibility follows by directly invoking the best-in assumption, i.e., $c_{i^*} \leq B$, for all $c_{i^*}$. \smallskip
 
\noindent\underline{Case 2:} $v_{i^*}(1) < \beta\OPTFk(\vec{c}_{-i^*})$. As in the proof of Lemma \ref{lemma:multi-bf}, we have that
\begin{align*}
 \sum_{i=1}^np_i(\vec{c})&= \sum_{i \in W(\vec{x}(\vec{c}))}p_i(\vec{c})= \sum_{i \in W(\vec{x}(\vec{c}))}\sum_{j=1}^{x_i(\vec{c})}p_{ij}(\vec{c}_{-i})\\
 &\leq \frac{B}{1-\alpha} \sum_{i \in W(\vec{x}(\vec{c}))} \sum_{j=1}^{x_i(\vec{c})} \frac{m_i(j)}{\OPTFk(\vec{c}_{-i})}
 =\frac{B}{1-\alpha} \sum_{i \in W(\vec{x}(\vec{c}))} \frac{v_i(x_i(\vec{c}))}{\OPTFk(\vec{c}_{-i})} \\
 &=\frac{B}{1-\alpha}\left( \sum_{i \in W(\vec{x}(\vec{c}))\setminus\{\ell\}} \frac{v_i(x_i(\vec{c}))}{\OPTFk(\vec{c}_{-i})} \ \ + \ \ 
 \frac{v_{\ell}(x_{\ell}(\vec{c})) - m_{\ell}(x_{\ell}(\vec{c}))}{\OPTFk(\vec{c}_{-\ell})} \ + \ \frac{m_{\ell}(x_{\ell}(\vec{c}))}{\OPTFk(\vec{c}_{-\ell})} \right)\\
 &\leq \frac{B}{1-\alpha}\left( (1+\beta k)\frac{\sum_{i \in W(\vec{x}(\vec{c}))} v_i(x_i(\vec{c})) - m_{\ell}(x_{\ell}(\vec{c}))}{\OPTFk(\vec{c})} + \frac{m_{\ell}(x_{\ell}(\vec{c}))}{\OPTFk(\vec{c}_{-\ell})} \right)\\
 &\leq \frac{B}{1-\alpha}\left( \alpha(1+\beta k)\frac{\sum_{i \in W(\vec{x}(\vec{c}))} v_i(x_i(\vec{c})) - m_{\ell}(x_{\ell}(\vec{c}))}{\V(\vec{x}(\vec{c})-m_{\ell}(x_{\ell}(\vec{c}))} + \frac{m_{\ell}(x_{\ell}(\vec{c}))}{\OPTFk(\vec{c}_{-\ell})} \right)\\
 &=\frac{B}{1-\alpha}\left(\alpha(1+\beta k) + \frac{m_{\ell}(x_{\ell}(\vec{c}))}{\OPTFk(\vec{c}_{-\ell})} \right)
 \leq \frac{B}{1-\alpha} \left(\alpha(1+\beta k) + \beta \right)\,.
\end{align*}

The first equality reflects the fact that $p_{ij}(\vec{c}_{-i}) = 0$ for $i \notin W(\vec{x}(\vec{c})$, whereas the second equality is due to an adaptation of Lemma \ref{lemma:payment-breakdown} for the best-in setting.\footnote{\ The definition of the critical payments for the best-in setting is completely analogous to the one for \MultiMechanism in Section \ref{sec:BF-Greedy} for the all-in setting.} The first inequality follows by applying Lemma \ref{lem:critical-upper-opt} for every agent $i \in W(\vec{x}(\vec{c}))$ and every $j\in \{1\dots, x_{i}(\vec{c})\}$. Then, the second inequality is due to Lemma \ref{lemma:best-in-opt-minus-i}, whereas the third inequality is a consequence of the \textbf{while} condition of \BestInMechanism on line \ref{line:BIM_while}. Finally, the last inequality is due to the fact that 
\[\frac{m_{\ell}(x_{\ell}(\vec{c}))}{\OPTFk(\vec{c}_{-\ell})} \leq \frac{v_{\ell}(1)}{\OPTFk(\vec{c}_{-\ell})} \leq \frac{v_{i^*}(1)}{\OPTFk(\vec{c}_{-i^*})} < \beta\,.\]
Budget-feasibility follows by observing that, since $\beta=\frac{1-2\alpha}{\alpha k +1}$, it holds that 

\begin{equation*}
 \frac{\alpha(1+\beta k) + \beta}{1-\alpha} = \frac{\alpha}{1-\alpha} + \beta \cdot \frac{\alpha k + 1}{1-\alpha} = \frac{\alpha}{1-\alpha} + \frac{1-2 \alpha}{1-\alpha}=1. 
\end{equation*}

This concludes the proof. 
\end{proof}



\section{Efficient Computation of Payments} \label{appendix:new-section}

In this section, we show how the payments returned by \MultiMechanism\ can be computed in polynomial time. 
For ease of reference, we repeat the necessary notation here. Given an instance $I=(N, \vec{c}, B, k, (v_i)_{i \in N})$ of the $k$-level model, $\vec{x}(\vec{c}) \in \{0, \dots, k\}^n$ is the allocation returned by \MultiMechanism. According to line 14 of the mechanism, the payments $\vec{p}(\vec{c}) = (p_1(\vec{c}), \dots, \allowbreak p_n(\vec{c}))$ are computed as defined in \eqref{eq:payment-id-BF} of Theorem \ref{thm:archerPayments}.

Recall from Section \ref{sec:BF-Greedy} that for each $i \in W(\vec{x}(\vec{c}))$ and $j = 1, \dots, x_i(\vec{c})$, $Q_{ij}(\vec{c}_{-i})$ is the set of all points $q \in \left[c_i, B/k\right]$ satisfying $\lim_{z \rightarrow q^-} x_i(z, \vec{c}_{-i}) \ge j$ and $\lim_{z \rightarrow q^+} x_i(z, \vec{c}_{-i}) \le j$. Furthermore, recall that $p_{ij}(\vec{c}_{-i})=\sup(Q_{ij}(\vec{c}_{-i}))$, the critical payment of the $j$th level of service of agent $i$.
By Lemma \ref{lemma:payment-breakdown}, we can express the total payment of each agent $i \in W(\vec{x}(\vec{c}))$ as the sum of these per-level critical payments, i.e.,
\[p_i(\vec{c}) = \sum_{r=1}^{x_i(\vec{c})} p_{ir}(\vec{c}_{-i})\,.\]

Now, for every  $i, j \in N$, let 
\begin{equation*}
\rho_i(\vec{c}, j) = v_{j}(k) \cdot \frac{\OPTFk(\vec{c}_{-i})}{\OPTFk(\vec{c}_{-j})}.   
\end{equation*}
For a profile $\vec{c}$, we use $i^*(\vec{c})$ to denote $i^*$ as computed in line \ref{line:i*} of \MultiMechanism. We distinguish the following two cases for any given agent $i \in W(\vec{x})$:

\smallskip
\noindent
\underline{Case 1:} $i = i^*(\vec{c})$ and $v_{i}(k) \geq \frac{\alpha}{1-\alpha} \OPTFk(\vec{c}_{-i})$. By line \ref{line:ifToCheckToPicki*}, the \textbf{if} part of \MultiMechanism is executed and we have that $x_i(\vec{c})=k$. Moreover, it is not hard to see that $p_{i1}(\vec{c}_{-i}) = \dots = p_{ik}(\vec{c}_{-i})$, as if agent $i$ were to unilaterally declare a cost $c_i' \in \left[c_i, B / k \right]$, they would still be hired for $k$ levels of service, as long as $i = i^*(c'_i, \vec{c}_{-i})$. 
Otherwise, whenever there exists a $c'_i \in [c_i, B / k]$ for which $i \neq i^*(c'_i, \vec{c}_{-i})$, we have that the \textbf{if} part of \MultiMechanism is still executed as $v_i(k) \leq \rho_i\left((c'_i, \vec{c}_{-i}),i^*(c'_i, \vec{c}_{-i})\right)$, and $x_i(c'_i, \vec{c}_{-i}) = 0$.\footnote{Note that this is the only possible alternative outcome. This is implied by Lemma \ref{lemma:monotone}.}
Therefore, for $r = 1, \dots, x_i(\vec{c})$, we can write
\begin{equation*}\label{eq:i-start-critical payments}
    p_{ir}(\vec{c}_{-i}) = \inf \left(\left\{\tfrac{B}{k}\right\}\cup \left\{c'_i \in \left[c_i, \tfrac{B}{k}\right] \mid  i \neq i^*(c'_i, \vec{c}_{-i}) \right\} \right).
\end{equation*}
If there is a $c'_i \in \left[c_i, B/k\right]$ such that $i \neq i^*(c'_i, \vec{c}_{-i}) = j $, we can compute it in polynomial time. To compute this $c'_i$, consider the non-increasing ordering of marginal value-per-cost ratios for the agents $N \backslash \{j\}$ with respect to $(\vec{c}_{-\{j,i\}}, c_i)$. We start by computing the \textit{costs of interest}. These are all the costs $b \in [c_i, B/k]$ for which there is a change in this ordering due to the unilateral deviation of agent $i$ to $b$. There are at most $(n-1)k^2$ such points in the interval $[c_i, B/k]$. 
For notational convenience, define 
\[
\theta_i(j,\vec{c}_{-i}) := 
v_i(j) \frac{\OPTFk(\vec{c}_{-i})}{v_i(k)}.
\]
Starting with the smallest cost of interest, denoted by $c^{1}_i$, in this interval, one can check if 
$\OPTFk(\vec{c}_{-\{j,i\}}, c^{1}_{i}) \le \theta_i(j,\vec{c}_{-i})$.
If this is true, one can find the point $c'_i$ where equality is reached in polynomial time. This is either at $c^{1}_i = c'_i$, or in $[c_i, c^{1}_i)$. For the latter, this point can be found by ordering the marginal values for $N \backslash \{j \}$ with respect to $(\vec{c}_{-\{i,j\}}, c_i)$, and finding the point where the first $s$ marginal values and a fraction of the $(s+1)$th marginal value add up to 
$\theta_i(j,\vec{c}_{-i})$.
One can then find $c'_i$ such that this allocation is exactly budget feasible. 
If 
$\OPTFk(\vec{c}_{-\{j,i\}}, c^{1}_{i}) > \theta_i(j,\vec{c}_{-i})$,
continue the same procedure with $c^{2}_i$ and $c^{1}_i$ instead of $c^{1}_i$ and $c_i$ respectively.


\smallskip

\noindent
\underline{Case 2:} $v_{i^*(\vec{c})}(k) < \frac{\alpha}{1-\alpha} \OPTFk(\vec{c}_{-i^*(\vec{c})})$. In this case, the \textbf{else} part of \MultiMechanism\ is executed. Recall the definition of $\ell(\vec{x}(\vec{c}))$ as the index in $W(\vec{x}(\vec{c}))$ of the agent hired for the least efficient level of service in $\vec{x}(\vec{c})$. For notational convenience, we define 
\[
\sigma(\vec{c}) := \frac{m_{\ell(\vec{x}(\vec{c}))}(x_i(\vec{c}))}{c_{\ell(\vec{x}(\vec{c}))}}.
\]
Fix a level of service $r \in [x_i(\vec{c})]$ for agent $i$. By the definition of $\ell(\vec{x}(\vec{c}))$ , we have $\frac{m_i(r)}{c_i} \geq \sigma(\vec{c})$.

Intuitively, there are two events under which agent $i$ will no longer be hired for their $r$th level of service when declaring a unilateral deviation $c'_i \in [c_i, B/k]$. The first event occurs when, under $(c'_i, \vec{c}_{-i})$, the $r$th level of service is no longer efficient enough to be selected by \MultiMechanism. Formally, $c'_i$ is such that $\frac{m_i(r)}{c'_i} < \sigma(c'_i, \vec{c}_{-i})$.
The second event occurs when, under $(c'_i, \vec{c}_{-i})$, the \textbf{if} condition in line \ref{line:ifToCheckToPicki*} becomes \textbf{true}, meaning:
\begin{equation*}
v_{i^*(c'_i, \vec{c}_{-i})}(k) \geq \frac{\alpha}{1-\alpha} \OPTFk\left(c'_i, \vec{c}_{-\{i, i^*(\vec{c})\}}\right)    
\end{equation*}
and agent $i$'s allocation under $(c'_i, \vec{c}_{-i})$ becomes 0.\footnote{Note that by the definition of this case, it cannot be that $i = i^*(c'_i, \vec{c}_{-i})$.}

Thus, the critical payment for agent $i$ for the $r$th level of service is the smallest value in $[c_i, B/k]$ at which at least one of these two events is triggered. If this is not the case, then the critical payment for this level is simply $B/k$. Formally:
\begin{equation*}
    p_{ir}(\vec{c}_{-i})= \inf \left(\left\{\frac{B}{k}\right\} \cup \left\{c'_i \in \left[c_i, \frac{B}{k}\right] \,\Big|\, \frac{m_i(r)}{c'_i} < \sigma(c'_i, \vec{c}_{-i})  \lor  v_{i^*(c'_i, \vec{c}_{-i})}(k) \geq \frac{\alpha}{1-\alpha} \OPTFk\left(c'_i, \vec{c}_{-\{i, i^*(\vec{c})\}}\right)\right\}     \right)
\end{equation*}

We now describe how the critical payment $p_{ir}(\vec{c}_{-i})$ is computed by focusing on each event separately.

To compute the smallest cost such that  ${m_i(r)}/{c'_i} < \sigma(c'_i, \vec{c}_{-i})$ if such a cost exists, it is sufficient to consider the points $c'_i \in [c_i ,B/k]$ for which the ratio ${m_i(r)}/{c'_i}$ equals the marginal value-per-cost ratio of less efficient levels of service according to the ordering of Algorithm \ref{algo:fk-knapsack}, which are at most $nk$. These will be the \textit{costs of interest} for this event. Consider these costs of interest in increasing order, and check if for any of these costs the \textbf{while} condition in line 11 is evaluates to \textbf{true}. Otherwise, suppose this happens for a cost $\hat{c}_i$ and let $\bar{c}_i$ be the cost considered before $\hat{c}_i$ (if this happens for the first cost of interest then $\bar{c}_i = c_i$). Then
\begin{equation*}
   p_{ir}(\vec{c}_{-i}) = \inf \left(  \left\{ \hat{c}_i \right\} \cup \left\{c'_i \in [\bar{c}_i, \hat{c}_i ) \,\Big|\,  (v(\vec{x}(\vec{c}_{-i},\bar{c}_{i})) - m_{i}(r)) \frac{1}{\alpha} = \OPTFk\left(\vec{c}_{-i}, c'_i \right)       \right\}     \right) ,
\end{equation*}
which can be computed in polynomial time. Note that in order to compute a $c'_i \in [\bar{c}_i, \hat{c}_i )$ in the set above (if it exists), one should again look into the costs for which the ordering of marginal value-per-cost ratios of Algorithm \ref{algo:fk-knapsack} changes. 

For the second event, we need to find the smallest $c'_i \in \left[c_i, B/K \right]$ such that the \textbf{if} condition in line 2 is \textbf{true} for $i^*(c'_i, \vec{c}_{-i})$ if such a cost exists. This can be done similarly to Case $1$. We can check for each agent $j \in N \setminus \{i\}$ if there exists such a $c'_i$ for which the \textbf{if} condition in line 2 is \textbf{true}, and then take the smallest $c'_i$. The difference is that instead of comparing with $\theta_i(j,\vec{c}_{-i})$, we compare with $v_j(k) (1 - \alpha)/\alpha$.

\end{document}